\documentclass[onecolumn, journal]{IEEEtran}  

\IEEEoverridecommandlockouts                              
\usepackage[utf8]{inputenc}
\usepackage{amsmath,bm}
\usepackage{amsfonts}
\usepackage{amssymb}
\usepackage{graphicx}
\usepackage{amsfonts}
\usepackage{hyperref}
\usepackage{caption}
\usepackage{subcaption}
\usepackage{algorithm}
\usepackage[noend]{algpseudocode}
\usepackage{xcolor}
\usepackage{cite}
\usepackage[normalem]{ulem}

\newtheorem{Theorem}{Theorem}

\newtheorem{Lemma}{Lemma}
\newtheorem{Problem}{Problem}

\newtheorem{Remark}{Remark}
\newtheorem{Assumption}{Assumption}
\usepackage{tikz}
	\usetikzlibrary{shadows,shapes,arrows}
	\tikzstyle{frame} = [draw, -latex]
	\tikzstyle{line} = [draw]
	\tikzstyle{line2} = [draw, dashdotted]
	\tikzstyle{line3} = [draw, dashed]
	\tikzstyle{line3UD} = [draw, dashed]
	\tikzstyle{place} = [circle, draw=black, fill=white, thick, inner sep=2pt, minimum size=1mm]
	\tikzstyle{place2} = [circle, draw=black, fill=black, thick, inner sep=2pt, minimum size=1mm]
	\tikzstyle{placeRed} = [circle, draw=red, fill=red, thick, inner sep=2pt, minimum size=1mm]
	\tikzstyle{vertex} = [circle, draw=black, fill=black, thick, inner sep=2pt, minimum size=1mm]
\usepackage{tkz-euclide}
\usepackage{multirow}

\makeatletter
\def\algbackskip{\hskip-\ALG@thistlm}
\makeatother
\usepackage[switch,pagewise]{lineno}
\allowdisplaybreaks

\title{\LARGE \bf Distributed least square approach for solving a multiagent linear algebraic equation}
\author{Viet Hoang Pham$^{1}$ and Hyo-Sung Ahn$^{1}$ 
\thanks{\small $^{1}$School of Mechanical Engineering, Gwangju Institute of Science and Technology, Gwangju, Korea. E-mails: {vietph@gist.ac.kr}; {hyosung@gist.ac.kr.}}
}

\begin{document}
\maketitle 
\thispagestyle{empty}
\pagestyle{empty}
\begin{abstract}
This paper considers a linear algebraic equation over a multiagent network. The coefficient matrix is partitioned into multiple blocks; each agent only knows a subset of these blocks in different row and column partitions. Based on a proximal ADMM algorithm, we design a distributed method for every agent to find its corresponding parts in one least square solution of the considered linear algebraic equation. Each agent uses only its information and communicates with its neighbors.
We show that the designed method achieves an exponentially fast convergence for an arbitrarily initial setup.
Numerical simulations in MATLAB are provided to verify the effectiveness of the designed method.
\end{abstract}
\section{Introduction}
The control and optimization problems of multiagent systems have become popular in many applications. The main reason is inherent complexity of large-scale problems and memory dispersion with the growth of big data. Depending on the communication network among agents, a multiagent system can have a hierarchical or distributed scheme. In a hierarchical setup, one master agent plays as a coordinator. This agent communicates with and/or assigns tasks to all others. In a distributed setup, agents work in parallel, and each agent can exchange information with only its neighbors. Because of the scalability and privacy protection requirements, distributed ones have become more desired recently.

Many fundamental problems in various engineering applications can be deduced into solving a linear algebraic equation $\textbf{H}\textbf{z} = \textbf{h}$ where $\textbf{H}$ and $\textbf{h}$ are coefficient matrix and vector. Various distributed methods have been developed to find the solution of $\textbf{H}\textbf{z} = \textbf{h}$ when each agent knows only a subset of rows in $\textbf{H}$ and $\textbf{h}$. The first approach considers constrained consensus problems \cite{ShaoshuaiMou2015, ShaoshuaiMou2016, OnurCihan2019, JiLiu2017, GoudongShi2017, JiLiu2018, JiLiu2016}. By using an orthogonal projection of the kernel of the local coefficient matrix, the estimated solution at every iteration of each agent always satisfies its linear equations. Distributed consensus schemes are exploited to drive estimated solutions of all agents to converge to an agreement. The methods in \cite{ShaoshuaiMou2015, ShaoshuaiMou2016, OnurCihan2019, JiLiu2017, JiLiu2018, GoudongShi2017} are discrete-time, and the one in \cite{JiLiu2016} is continuous-time. The convergence of these methods is guaranteed if the considered linear algebraic equation has at least one exact solution.

Another approach considers the problem of solving a distributed linear algebraic equation as finding the optimal solution of a multiagent optimization problem \cite{GoudongShi2017, XuanWang2019, YangLiu2019, TaoYang2020, MohammadJahvani2022, PengWang2022}. The local cost function in each agent is the total square error of its local linear equations, and the constraints are added to force the agreement for estimated solutions among neighboring agents. Various distributed optimization algorithms, such as subgradient-based methods \cite{BahmanGharesifard2014, GuannanQu2017}, and ADMM \cite{BingshengHe2015, XiaoweiPan2022}, can be applied to solve the optimization problem. When the considered linear algebraic equation has one or several exact solutions, all agents can asymptotically determine one of such points \cite{PengWang2022}. If the exact solution does not exist, the final converged solution is one least square solution, which minimizes the total square errors, of the considered linear algebraic equation \cite{GoudongShi2017, XuanWang2019, YangLiu2019, TaoYang2020, MohammadJahvani2022}.

In the works mentioned above \cite{ShaoshuaiMou2015, ShaoshuaiMou2016, OnurCihan2019, JiLiu2017, GoudongShi2017, JiLiu2018, JiLiu2016, PengWang2022, XuanWang2019, YangLiu2019, TaoYang2020, MohammadJahvani2022}, the authors consider a distributed linear algebraic equation $\textbf{H}\textbf{z} = \textbf{h}$ assuming that the coefficient matrix $\textbf{H}$ is divided into multiple row partitions and each agent knows one of these partitions. Consequently, in their proposed algorithm, except for the one in \cite{ShaoshuaiMou2016}, every agent controls a state vector having the same dimension as the unknown vector $\textbf{z}$ in its update. In \cite{ShaoshuaiMou2016}, the dimension of the state vector controlled by one agent can be reduced when considering the sparsity in its row partition.
Different from \cite{ShaoshuaiMou2015, ShaoshuaiMou2016, OnurCihan2019, JiLiu2017, GoudongShi2017, JiLiu2018, JiLiu2016, PengWang2022, XuanWang2019, YangLiu2019, TaoYang2020, MohammadJahvani2022}, the authors in \cite{XuanWang2020} assume that the coefficient matrix $\textbf{H}$ is first divided into either multiple row partitions or column partitions. Then, each of these partitions is further divided into multiple blocks. Only one agent can access each block. These setups are more suitable than those in \cite{ShaoshuaiMou2015, ShaoshuaiMou2016, OnurCihan2019, JiLiu2017, GoudongShi2017, JiLiu2018, JiLiu2016, PengWang2022, XuanWang2019, YangLiu2019, TaoYang2020, MohammadJahvani2022} when the coefficient matrix $\textbf{H}$ has large rows and columns. However, the algorithms proposed in \cite{XuanWang2020} require a hierarchical structure. Each agent is in only one cluster corresponding to one row or column partition. Each cluster has one aggregator to collect information on agents in this cluster and exchange information with aggregators of others.

This paper considers a linear algebraic equation $\textbf{H}\textbf{z} = \textbf{h}$ when the rows and columns of the coefficient matrix $\textbf{H}$ are distributed over a network of agents. Let the matrix $\textbf{H}$ be divided into multiple blocks. Each agent only knows a subset of the divided blocks, which can be located in many row and column partitions.
Unlike \cite{XuanWang2020}, we do not require an aggregator to collect information from agents corresponding to the same row or column partition. For the row partitions whose data are dispensed across multiple agents, we introduce some virtual variables to transform their linear equations into local constraints of individual agents combining coupling constraints among neighboring agents. Consequently, we obtain a distributed optimization problem whose optimal solution is one least square solution of the considered linear algebraic equation. Then we apply a proximal ADMM algorithm \cite{BingshengHe2015} to design a distributed method for finding the optimal solution of the obtained distributed optimization problem. Each agent must only use its information and communicate with its neighbors.
Since the designed method is discrete-time, the asymptotic convergence of the ADMM-based method implies the exponential convergence of the estimated solution to the optimal one.
In addition, our designed method has no requirements for step-size selection or initial point setup.

The remainder of the paper is organized as follows. Section II provides some preliminaries for the analysis in later sections. In Section III, we introduce distributed linear algebraic equation setup and optimization-based approach. Then we formulate a distributed optimization problem whose optimal solution can be used to determine the least square solution of the considered linear algebraic equation in Section IV. Based on a proximal ADMM algorithm \cite{BingshengHe2015}, Section V designs a distributed method for solving the optimization problem formulated in Section IV. Numerical simulations are provided in Section VI to illustrate the effectiveness of the designed method. Section VII concludes this paper.
\section{Preliminaries}
\subsection{Notations}
We use $\mathbb{R}, \mathbb{R}^n$, and $\mathbb{R}^{m \times n}$ to denote the set of real numbers, the set of $n$-dimension real vectors, and the set of real matrices having $m$ rows and $n$ columns.
Denote by $\textbf{1}_n$ and $\textbf{0}_n$ the vectors in $\mathbb{R}^{n}$ whose all elements are $1$ and $0$, respectively.
Let $\textbf{I}_n$ be the identity matrix in $\mathbb{R}^{n \times n}$ and $\textbf{0}_{m \times n}$ be the zero matrix in $\mathbb{R}^{m \times n}$. When the dimensions are clear, the subscripts of the matrices $\textbf{I}_n, \textbf{0}_{m \times n}$, and the vector $\textbf{1}_n, \textbf{0}_n$ can be removed.
Denote by $\textbf{H}^T$ the transpose matrix of $\textbf{H}$.
We write $\textbf{H} \succ 0$ ($\succeq 0$) means $\textbf{H}$ is a positive (semi-)definite matrix.
Define $||\textbf{x}||_{\textbf{H}}^2 = \textbf{x}^T\textbf{H}\textbf{x}$. When $\textbf{H} = \textbf{I}$, we write it as $||\textbf{x}||^2$ for simplicity. If $x$ is a scalar, we write $|x| = ||x||$.
We use $|\mathbb{S}|$ to denote the cardinality of the set $\mathbb{S}$.

Let $\mathcal{H} = \left\{\textbf{h}_1, \textbf{h}_2, \dots, \textbf{h}_m\right\}$ be the list of $m$ vectors, we define the column vector
\[\textrm{col } \mathcal{H} = \textrm{col} \left\{\textbf{h}_1, \textbf{h}_2, \dots, \textbf{h}_m\right\} = \left[\textbf{h}_1^T, \textbf{h}_2^T, \dots, \textbf{h}_m^T\right]^T.\]
For a list of matrices, $\mathcal{H} = \left\{\textbf{H}_1, \textbf{H}_2, \dots, \textbf{H}_m\right\}$, we use $\textrm{blkdiag }\mathcal{H}$ to denote the block-diagonal matrix whose main blocks are matrices in the set $\mathcal{H}$. In addition, we define two matrices
\[\textrm{blkcol }\mathcal{H} = \textrm{blkcol}\left\{\textbf{H}_1, \textbf{H}_2, \dots, \textbf{H}_m\right\} = \left[\begin{matrix}\textbf{H}_1^T & \textbf{H}_2^T & \cdots & \textbf{H}_m^T\end{matrix}\right]^T,\]
\[\textrm{blkrow }\mathcal{H} = \textrm{blkrow}\left\{\textbf{H}_1, \textbf{H}_2, \dots, \textbf{H}_m\right\} = \left[\begin{matrix}\textbf{H}_1 & \textbf{H}_2 & \cdots & \textbf{H}_m\end{matrix}\right].\]
\subsection{Proximal ADMM Algorithm}
Consider the following convex optimization problem
\begin{equation}\label{eq_ADMMproblem}
\min\limits_{\textbf{x} \in \mathcal{X}, \textbf{y} \in \mathcal{Y}} \textrm{ }\Psi_x(\textbf{x}) + \Psi_y(\textbf{y}) \textrm{ s.t. } \textbf{F}_x\textbf{x} + \textbf{F}_y\textbf{y} = \textbf{f}.
\end{equation}
where $\Psi_x(\textbf{x}), \Psi_y(\textbf{y})$ are convex functions, $\mathcal{X}, \mathcal{Y}$ are convex sets, $\textbf{F}_x, \textbf{F}_y$ and $\textbf{f}$ are known matrices and vector.
The augmented Lagrangian function associated to the problem \eqref{eq_ADMMproblem} is defined as \[\mathcal{L}_{\rho}(\textbf{x},\textbf{y},\boldsymbol{\lambda}) = \Psi_x(\textbf{x}) + \Psi_y(\textbf{y}) + \frac{\rho}{2}\Bigl|\Bigl|\textbf{F}_x\textbf{x} + \textbf{F}_y\textbf{y} - \textbf{f} - \frac{1}{\rho}\boldsymbol{\lambda}\Bigr|\Bigr|^2\] where $\boldsymbol{\lambda}$ is the dual variable corresponding to the equality constraint in \eqref{eq_ADMMproblem} and $\rho > 0$ is a positive penalty parameter.
The proximal ADMM algorithm proposed in \cite{BingshengHe2015} for solving the convex optimization problem \eqref{eq_ADMMproblem} is given as follows.
\begin{subequations}\label{eq_ADMM}
\begin{align}
\textbf{x}(s+1) &= \arg\min\limits_{\textbf{x} \in \mathcal{X}} \Bigl\{\mathcal{L}_{\rho}(\textbf{x},\textbf{y}(s),\boldsymbol{\lambda}(s)) + \frac{1}{2}\left|\left|\textbf{x} - \textbf{x}(s)\right|\right|_{\textbf{G}}^2 \Bigr\},\\
\textbf{y}(s+1) &= \arg\min\limits_{\textbf{y} \in \mathcal{Y}} \mathcal{L}_{\rho}(\textbf{x}(s+1),\textbf{y},\boldsymbol{\lambda}(s)),\\
\boldsymbol{\lambda}(s+1) &= \boldsymbol{\lambda}(s) - \rho \Bigl(\textbf{F}_x\textbf{x}(s+1) + \textbf{F}_y\textbf{y}(s+1) - \textbf{f}\Bigr),
\end{align}
\end{subequations}
where $\textbf{G}$ is a symmetric and positive semidefinite matrix, and the initial point $(\textbf{x}(0), \textbf{y}(0), \boldsymbol{\lambda}(0))$ can be chosen arbitrarily such that $\textbf{x}(0) \in \mathcal{X}, \textbf{y}(0) \in \mathcal{Y}$.
The following theorem is Lemma 1 in \cite{XiaoweiPan2022}, which is based on the convergence result of the proximal ADMM algorithm \eqref{eq_ADMM} in \cite{BingshengHe2015}.
\begin{Theorem}\label{th_ADMM}
Let $\boldsymbol{\Theta} = blkdiag\left\{ \textbf{G}, \rho\textbf{F}_y^T\textbf{F}_y, \frac{1}{\rho}\textbf{I} \right\}$ and $\textbf{w}(s) = [\textbf{x}(s)^T, \textbf{y}(s)^T, \boldsymbol{\lambda}(s)^T]^T$. Then we have
\begin{enumerate}
\item $||\textbf{w}(s) - \textbf{w}(s+1)||_{\boldsymbol{\Theta}}^2 = o\left(\frac{1}{s}\right)$
\item $|| \textbf{F}_x\textbf{x}(s) + \textbf{F}_y\textbf{y}(s) - \textbf{f}|| = o\left(\frac{1}{\sqrt{s}}\right)$
\item $|\Psi_x(\textbf{x}(s)) + \Psi_y(\textbf{y}(s)) - \psi^{opt}| = o\left(\frac{1}{\sqrt{s}}\right)$ where $\psi^{opt}$ is the optimal value of the cost function in \eqref{eq_ADMMproblem}.
\end{enumerate}
\end{Theorem}
\section{Problem Formulation}
\subsection{Distributed linear algebraic equation}
Consider a linear algebraic equation
\begin{equation}\label{eq_LE}
\textbf{H}\textbf{z} = \textbf{h},
\end{equation}
where $\textbf{z} \in \mathbb{R}^{n}$ is unknown vector, $\textbf{H} \in \mathbb{R}^{m \times n}$ and $\textbf{h} \in \mathbb{R}^{m}$ are coefficients matrix and vector, respectively.
Let $\textbf{z}^{opt}$ be one least square solution to \eqref{eq_LE}. Then $\textbf{z}^{opt}$ is one optimal solution to the following optimization problem.
\begin{equation}\label{eq_LE_solution}
\min_{\textbf{z}} \frac{1}{2}\bigl|\bigl|\textbf{H}\textbf{z} - \textbf{h}\bigr|\bigr|^2.
\end{equation}
In general, the linear algebraic equation \eqref{eq_LE} can have one or more least square solutions.
Define $\Psi^{opt} = \frac{1}{2}||\textbf{H}\textbf{z}^{opt} - \textbf{h}||^2$.
We have $\Psi^{opt} = 0$ if the linear algebraic equation \eqref{eq_LE} has at least one exact solution. If $\textbf{H}$ is a full column rank matrix, we have $\textbf{z}^{opt} = (\textbf{H}^T\textbf{H})^{-1}\textbf{H}^T\textbf{h}$ is unique and $\Psi^{opt} = \frac{1}{2}\textbf{h}^T\textbf{h} - \frac{1}{2}\textbf{h}^T\textbf{H}(\textbf{H}^T\textbf{H})^{-1}\textbf{H}^T\textbf{h}$.

In this paper, we are interested in determining one least square solution to the linear algebraic equation \eqref{eq_LE} when the coefficient matrix $\textbf{H}$ can be subdivided into multiple blocks $\textbf{H}_{kl}$, where $1 \le k \le K$ and $1 \le l \le L$, as in \eqref{eq_LE_partitions} and each nonzero block is known by only one agent in a network of $N$ agents.
\begin{equation}\label{eq_LE_partitions}
\textbf{H} = \left[\begin{matrix}
\textbf{H}_{11} & \textbf{H}_{12} & \cdots & \textbf{H}_{1L}\\
\textbf{H}_{21} & \textbf{H}_{22} & \cdots & \textbf{H}_{2L}\\
\vdots & \vdots & \ddots & \vdots\\
\textbf{H}_{K1} & \textbf{H}_{K2} & \cdots & \textbf{H}_{KL}
\end{matrix}\right].
\end{equation}
The coefficient matrix $\textbf{H}$ in \eqref{eq_LE_partitions} includes $KL$ blocks corresponding to $K$ row and $L$ column partitions.
As each agent $i$, where $1 \le i \le N$, knows some nonzero blocks, we have $KL \ge N$ in general setting. We define $\mathcal{A}[\textbf{H}_{kl}]$ as an operator to express the agent knowing the block $\textbf{H}_{kl}$. Since there may be some zero blocks in $\textbf{H}$, we conventionally define $\mathcal{A}[\textbf{H}_{kl}] = 0$ if $\textbf{H}_{kl} = \textbf{0}$. Note that $0$ is not an agent.
For each agent $i$, we define the relation between the block $\textbf{H}_{kl}$ and this agent by $\delta_i[\textbf{H}_{kl}]$ where $\delta_i[\textbf{H}_{kl}] = 1$ if $\mathcal{A}[\textbf{H}_{kl}] = i$ and $\delta_i[\textbf{H}_{kl}] = 0$, otherwise.
Let $\mathcal{B}_k^{(i)} = \{\textbf{H}_{kl}: \delta_i[\textbf{H}_{kl}] = 1, \forall 1 \le l \le L\}$ be the set of blocks in the row partition $k$ known by the agent $i$, where $1 \le k \le K$ and $1 \le i \le N$.
 
For each row partition $k$, where $1 \le k \le K$, of the coefficient matrix $\textbf{H}$, we define the index list of agents that have information of blocks in this row partition as $\mathcal{R}_k = \{\mathcal{A}[\textbf{H}_{k1}], \mathcal{A}[\textbf{H}_{k2}], \dots, \mathcal{A}[\textbf{H}_{kL}]\}$. Let $\mathcal{C}^l$ be the index list of agents that have information of blocks in the column partition $l$, i.e., $\mathcal{C}^l = \{\mathcal{A}[\textbf{H}_{1l}], \mathcal{A}[\textbf{H}_{2l}], \dots, \mathcal{A}[\textbf{H}_{Kl}]\}$, where $1 \le l \le L$.
As $\mathcal{R}_k$ and $\mathcal{C}^l$, where $1 \le k \le K$ and $1 \le l \le L$, are lists, they can contain duplicated indexes. We define the set $\mathbb{S}(\mathcal{R}_k)$ derived from the list $\mathcal{R}_k$ by removing the index $0$ and merging duplicated indexes into one element. Similarly, we define $\mathbb{S}(\mathcal{C}^l)$ for the list $\mathcal{C}^l$ where $1 \le l \le L$.
We can obtain an equivalent linear algebraic equation to \eqref{eq_LE} if we exchange rows in the matrix $\textbf{H}$ and the vector $\textbf{h}$. So, the rows in $\textbf{H}$ can be interchanged such that the ones dispensed across the same list of agents are grouped in the same row partition.
Without loss of generality, we assume that $\mathcal{R}_k \neq \mathcal{R}_{k'}$ if $k \neq k'$ for all $ 1 \le k,k' \le K$.
\subsection{Optimization approach}
Denote by $m_k, n_l$ the dimensions of the block $\textbf{H}_{kl} \in \mathbb{R}^{m_k \times n_l}$, where $1 \le k \le K$ and $1 \le l \le L$. We have $m = \sum_{k = 1}^{K}m_k$ and $n = \sum_{l = 1}^{L}n_l$.
According to the partition of the coefficient matrix $\textbf{H}$, the unknown vector $\textbf{z}$ and the coefficient vector $\textbf{h}$ can be partitioned as $\textbf{z} = \textrm{col}\{\textbf{z}_{1}, \textbf{z}_{2}, \dots, \textbf{z}_{L}\}$ and $\textbf{h} = \textrm{col}\{\textbf{h}_{1}, \textbf{h}_{2}, \dots, \textbf{h}_{K}\}$ where
\begin{align*}
\textbf{z}_{l} &= \textrm{blkrow}\bigl\{\textbf{0}_{n_l \times \sum_{j = 0}^{l-1}n_j}, \textbf{I}_{n_l}, \textbf{0}_{n_l \times \bigl(n-\sum_{j = 0}^{l}n_j\bigr)}\bigr\}\textbf{z},\\
\textbf{h}_{k} &= \textrm{blkrow}\bigl\{\textbf{0}_{m_k \times \sum_{j = 0}^{k-1}m_j}, \textbf{I}_{m_k}, \textbf{0}_{m_k \times \bigl(m-\sum_{j = 0}^{k}m_j\bigr)}\bigr\}\textbf{h}.
\end{align*}
Here we set $m_0 = n_0 = 0$ conventionally.
Assume that each agent $i \in \mathbb{S}(\mathcal{C}^l)$ has a copy $\textbf{z}_l^{(i)}$ for the part $\textbf{z}_l$ of the unknown vector $\textbf{z}$. We have
\begin{equation}\label{eq_converted_temp}
\sum_{l = 1}^{L} \textbf{H}_{kl} \textbf{z}_l = \sum_{l = 1, \textbf{H}_{kl} \neq \textbf{0}}^{L} \textbf{H}_{kl} \textbf{z}_l^{(\mathcal{A}[\textbf{H}_{kl}])}, \forall 1 \le k \le K.
\end{equation}
If $\mathcal{A}[\textbf{H}_{kl}] \neq 0$, the vector $\textbf{H}_{kl} \textbf{z}_l^{\mathcal{A}[\textbf{H}_{kl}]}$ is private information of the agent $\mathcal{A}[\textbf{H}_{kl}]$.
To rewrite the error of the linear equations corresponding to the row partition $k$ as the sum of local private functions, we make the following assumption:
\begin{Assumption}\label{aspt_partitioned_data}
For every row partition $k$, where $1 \le k \le K$, each agent $i \in \mathbb{S}(\mathcal{R}_k)$ can access to a part of $\textbf{h}_k$ (i.e., can have a part of information of $\textbf{h}_k$). Let us denote the part of information as $\textbf{h}_{i,k}$. However, the combination of $\textbf{h}_{i,k}$ is equivalent to $\textbf{h}_k$ such that $\sum\limits_{i \in \mathbb{S}(\mathcal{R}_k)} \textbf{h}_{i,k} = \textbf{h}_k$.
\end{Assumption}
Under Assumption \ref{aspt_partitioned_data}, we have
\begin{equation}\label{eq_aspt_coupled_equations}
\sum_{l = 1}^{L} \textbf{H}_{kl} \textbf{z}_l - \textbf{h}_k = \sum_{i \in \mathbb{S}(\mathcal{R}_k)}\Bigl(\sum_{\textbf{H}_{kl} \in \mathcal{B}_k^{(i)}} \textbf{H}_{kl}\textbf{z}_l^{(i)} - \textbf{h}_{i,k}\Bigr),
\end{equation}
for all $1 \le k \le K$.
So, we have an equivalent problem of the optimization problem \eqref{eq_LE_solution} as follows.
\begin{subequations}\label{eq_LE_optimizationproblem}
\begin{align}
&\min_{\textbf{z}_l^{(i)}} \frac{1}{2} \sum_{k = 1}^{K} \left|\left|\sum_{i \in \mathbb{S}(\mathcal{R}_k)}\left(\sum_{\textbf{H}_{kl} \in \mathcal{B}_k^{(i)}} \textbf{H}_{kl}\textbf{z}_l^{(i)} - \textbf{h}_{i,k}\right)\right|\right|^2\\
&\textrm{s.t. } \textbf{z}_l^{(i)} = \textbf{z}_l^{(j)}, \forall i,j \in \mathbb{S}(\mathcal{C}^l), i \neq j, \forall 1 \le l \le L. 
\end{align}
\end{subequations}
It is clear that the optimal cost value $\Phi^{opt} = \frac{1}{2}||\textbf{H}\textbf{z}^{opt} - \textbf{h}||^2$ is also the optimal value for the cost function in the problem \eqref{eq_LE_optimizationproblem}. In addition, each optimal solution to the optimization problem \eqref{eq_LE_optimizationproblem} can be used to determine one least square solution \eqref{eq_LE_solution} of the linear algebraic equation \eqref{eq_LE}.
\subsection{Communication graph}
In this paper, we use an undirected graph $\mathcal{G} = (\mathcal{V}, \mathcal{E})$ to illustrate the communication network among agents. In which, $\mathcal{V} = \{1, 2, \dots\}$ is the node set and $\mathcal{E} \subset \mathcal{V} \times \mathcal{V}$ is the edge set.
Each node $i \in \mathcal{V} = \{1, 2, \dots, N\}$ corresponds to an agent while an edge $(i,j) \in \mathcal{E}$ represents that the agent $i$ can exchange information with the agent $j$.
As $\mathcal{G}$ is an undirected graph, $(i,j) \in \mathcal{E}$ implies $(j,i) \in \mathcal{E}$ for every pair of nodes $i,j \in \mathcal{V}$.
The neighbor set of the agent $i \in \mathcal{V}$ is denoted by $\mathcal{N}_i = \left\{j \in \mathcal{V}: (i,j) \in \mathcal{E}\right\}$.
Since agents need to cooperate to find the optimal solution of the problem \eqref{eq_LE_optimizationproblem}, the communication graph $\mathcal{G}$ is required to be connected. That means there is at least one path from $i$ to $j$ for any pair $i,j \in \mathcal{V}$, i.e., there is a sequence of edges $(i_1,i_2), (i_2,i_3), \dots, (i_{k-1},i_k)$ such that $i_1 = i, i_k = j$ and $i_l \in \mathcal{V}, (i_{l-1},i_l) \in \mathcal{E}, \forall 2 \le l \le k$.

Let $\tilde{\mathcal{V}}$ be a node subset, i.e., $\tilde{\mathcal{V}} \subset \mathcal{V}$. The subgraph induced by $\tilde{\mathcal{V}}$ is defined as $\tilde{\mathcal{G}} = (\tilde{\mathcal{V}}, \tilde{\mathcal{E}})$ where $\tilde{\mathcal{E}} = \{(i,j) \in \mathcal{E}: i,j \in \tilde{\mathcal{V}}\}$.
We define $\mathcal{G}_{k} = (\mathbb{S}(\mathcal{R}_k), \mathcal{E}_{k})$ as the subgraph induced by the node set $\mathbb{S}(\mathcal{R}_k)$ for each row partition $k$, where $1 \le k \le K$. 
Similarly, we define $\mathcal{G}^{l} = (\mathbb{S}(\mathcal{C}^{l}), \mathcal{E}^{l})$ for each column partition $l$, where $1 \le l \le L$.
The following assumption is necessary.
\begin{Assumption}\label{aspt_topology}
The communication graph satisfies that
\begin{enumerate}
\item $\mathcal{G}^{l}$ is a connected graph for all $1 \le l \le L$,
\item $\mathcal{G}_{k}$ is a connected graph for all $1 \le k \le K$.
\end{enumerate}
\end{Assumption}
The first item in Assumption \ref{aspt_topology} is to allow agents in the set $\mathbb{S}(\mathcal{C}^{l})$ can cooperatively achieve an agreement on the part $\textbf{z}_l$ of the unknown vector $\textbf{z}$. Under this assumption, the constraints in (\ref{eq_LE_optimizationproblem}b) can be converted into coupling constraints among neighboring agents as follows.
\begin{equation}\label{eq_coupling_commonvariable_temp}
\textbf{z}_l^{(i)} = \textbf{z}_l^{(j)}, \forall j \in \mathcal{N}_i \cap \mathbb{S}(\mathcal{C}^l), \forall i \in \mathbb{S}(\mathcal{C}^l),
\end{equation}
for all $1 \le l \le L$.
When $i \in \mathbb{S}(\mathcal{R}_k)$, the agent $i$ takes part in the linear equations corresponding to the row partition $k$ as shown in \eqref{eq_aspt_coupled_equations}.
In general, there may be some row partitions $1 \le k \le K$ where $\mathbb{S}(\mathcal{R}_k) \neq \mathcal{N}_i$ for all $i \in \mathbb{S}(\mathcal{R}_k)$.
Since communication is only available between two neighboring agents, it is difficult for one agent to collect complete information on the linear equations corresponding to these row partitions. The second item in Assumption \ref{aspt_topology} facilitates a transformation of these linear equations into local linear equations of individual agents combining coupling constraints among neighboring agents, as presented later in Section III.
Summarizing this section, we state the main problem of this paper as follows.
\begin{Problem}
Assume that the communication graph of the considered multiagent network satisfies Assumption \ref{aspt_topology}.
Design a distributed method for each agent to find its corresponding parts in one optimal solution of the constrained optimization problem \eqref{eq_LE_optimizationproblem} while using only its information and communicating with its neighbors in $\mathcal{N}_i$.
\end{Problem}
\section{Distributed optimization reformulation}
Let $\bar{\textbf{z}}_i = \textrm{col}\{\textbf{z}_l^{(i)}: i \in \mathbb{S}(\mathcal{C}^l), \forall 1 \le l \le L\}$ be the vector of local variables of the agent $i \in \mathcal{V}$. For each neighboring agent $j \in \mathcal{N}_i$, define the matrix $\textbf{P}_{ij}$ such that $\textbf{P}_{ij}\bar{\textbf{z}}_i = \textrm{col}\{\textbf{z}_l^{(i)}: i,j \in \mathbb{S}(\mathcal{C}^l), \forall 1 \le l \le L\}$. The constraints in \eqref{eq_coupling_commonvariable_temp} are equivalent to the following one:
\begin{equation}\label{eq_coupling_commonvariable}
\textbf{P}_{ij}\bar{\textbf{z}}_i = \textbf{P}_{ji}\bar{\textbf{z}}_j, \forall (i,j) \in \mathcal{E}.
\end{equation}
Let $\Psi(\bar{\textbf{z}}_1, \bar{\textbf{z}}_2, \dots, \bar{\textbf{z}}_N)$ be the cost function in (\ref{eq_LE_optimizationproblem}a). The optimization problem \eqref{eq_LE_optimizationproblem} can be rewritten as
\begin{equation}\label{eq_LE_optimizationproblem_rewritten}
\min_{\bar{\textbf{z}}_i, \forall i \in \mathcal{V}} \Psi(\bar{\textbf{z}}_1, \bar{\textbf{z}}_2, \dots, \bar{\textbf{z}}_N) \textrm{ s.t. } \eqref{eq_coupling_commonvariable}.
\end{equation}
Consider a row partition $k$ where there is only one agent $i$ in the set $\mathbb{S}(\mathcal{R}_k)$, i.e., $\mathbb{S}(\mathcal{R}_k) = \{i\}$. Under Assumption \ref{aspt_partitioned_data}, the agent $i$ knows the vector $\textbf{h}_k$. In this case, we have
\begin{equation}\label{eq_local_equation_full}
\sum_{l = 1}^{L} \textbf{H}_{kl} \textbf{z}_l - \textbf{h}_k = \bar{\textbf{A}}_{i,k}\bar{\textbf{z}}_i - \bar{\textbf{a}}_{i,k}, \textrm{ if } \mathbb{S}(\mathcal{R}_k)=\{i\},
\end{equation}
where $\bar{\textbf{a}}_{i,k} = \textbf{h}_k$ and $\bar{\textbf{A}}_{i,k} = \textrm{blkrow}\Bigl\{\textbf{H}_{kl}\delta_i[\textbf{H}_{kl}]: i \in \mathbb{S}(\mathcal{R}_k) \cap \mathbb{S}(\mathcal{C}^l), \forall 1 \le l \le L\Bigr\}$.
Define by $\mathbb{M} = \{\epsilon_1, \epsilon_2, \dots, \epsilon_M\}$ the set of row partitions whose corresponding sets consist of two or more agents, i.e., $|\mathbb{S}(\mathcal{R}_{\epsilon_k})| \ge 2$ for all $\epsilon_k \in \mathbb{M}$.
For each row partition $\epsilon_k \in \mathbb{M}$ and agent $i \in \mathcal{V}$, define the matrix $\bar{\textbf{B}}_{i,\epsilon_k}$ by $\bar{\textbf{B}}_{i,\epsilon_k} = \textrm{blkrow}\Bigl\{\textbf{H}_{\epsilon_kl}\delta_i[\textbf{H}_{\epsilon_kl}]: i \in \mathbb{S}(\mathcal{R}_{\epsilon_k}) \cap \mathbb{S}(\mathcal{C}^l), \forall \le l \le L\Bigr\}$ and $\bar{\textbf{b}}_{i,\epsilon_k} = \textbf{h}_{i,\epsilon_k}$.
We have
\begin{equation}\label{eq_coupling_equation}
\sum_{l = 1}^{L} \textbf{H}_{\epsilon_kl} \textbf{z}_l - \textbf{h}_{\epsilon_k} = \sum_{i \in \mathbb{S}(\mathcal{R}_{\epsilon_k})} \left(\bar{\textbf{B}}_{i,\epsilon_k}\bar{\textbf{z}}_i - \bar{\textbf{b}}_{i,\epsilon_k}\right).
\end{equation}
Let $\textbf{u}_{\epsilon_k}$ be the error of the coupling linear equations in the row partition $\epsilon_k \in \mathbb{M}$, i.e.,
\begin{equation}\label{eq_LE_temp}
\textbf{u}_{\epsilon_k} = \sum_{i \in \mathbb{S}(\mathcal{R}_{\epsilon_k})} \left(\bar{\textbf{B}}_{i,\epsilon_k}\bar{\textbf{z}}_i - \bar{\textbf{b}}_{i,\epsilon_k}\right), \forall \epsilon_k \in \mathbb{M}.
\end{equation}
For every $\epsilon_k \in \mathbb{M}$, we define the vector $\textbf{v}_{ij, \epsilon_k}$ for each edge $(i,j) \in \mathcal{E}_{\epsilon_k}$ such that
\begin{equation}\label{eq_LE_virtual_var}
\frac{1}{|\mathbb{S}(\mathcal{R}_{\epsilon_k})|}\textbf{u}_{\epsilon_k} = \bar{\textbf{B}}_{i,\epsilon_k}\bar{\textbf{z}}_i - \bar{\textbf{b}}_{i,\epsilon_k} + \sum_{j \in \mathcal{N}_{i,\epsilon_k}} \left(\textbf{v}_{ij, \epsilon_k} - \textbf{v}_{ji, \epsilon_k}\right),
\end{equation}
for all $i \in \mathbb{S}(\mathcal{R}_{\epsilon_k})$, where $\mathcal{N}_{i,\epsilon_k} \equiv \mathcal{N}_i \cap \mathbb{S}(\mathcal{R}_{\epsilon_k})$.
We remark here that there always exists a set $\{\textbf{v}_{ij, \epsilon_k}: (i,j) \in \mathcal{E}_{\epsilon_k}\}$ that satisfying \eqref{eq_LE_virtual_var}.
Indeed, from \eqref{eq_LE_virtual_var}, we have
\begin{equation}\label{eq_LE_temp_virtual_var}
\sum\limits_{j \in \mathcal{N}_{i,\epsilon_k}} \left(\textbf{v}_{ij, \epsilon_k} - \textbf{v}_{ji, \epsilon_k}\right) = \frac{1}{|\mathbb{S}(\mathcal{R}_{\epsilon_k})|}\textbf{u}_{\epsilon_k} - \bar{\textbf{B}}_{i,\epsilon_k}\bar{\textbf{z}}_i + \bar{\textbf{b}}_{i,\epsilon_k},
\end{equation}
for all $i \in \mathbb{S}(\mathcal{R}_{\epsilon_k})$.
Consider \eqref{eq_LE_temp_virtual_var} as a system of linear equations with $\{\textbf{v}_{ij,\epsilon_k}: (i,j) \in \mathcal{E}_{\epsilon_k}\}$ is the set of unknown variables. It is clear that there are $|\mathbb{S}(\mathcal{R}_{\epsilon_k})|$ linear equations and $2|\mathcal{E}_{\epsilon_k}|$ variables. Since $\mathcal{G}_{\epsilon_k} = (\mathbb{S}(\mathcal{R}_{\epsilon_k}), \mathcal{E}_{\epsilon_k})$ is a connected graph, we have $|\mathcal{E}_{\epsilon_k}| \ge |\mathbb{S}(\mathcal{R}_{\epsilon_k})|$. This implies that the linear algebraic equation  \eqref{eq_LE_temp_virtual_var} has infinite roots.

Let $\textbf{v}_{ij, \epsilon_k}^{(i)}$ be the copy of $\textbf{v}_{ij, \epsilon_k} - \textbf{v}_{ji, \epsilon_k}, \forall (i,j) \in \mathcal{E}_{\epsilon_k},$ in the agent $i \in \mathbb{S}(\mathcal{R}_{\epsilon_k})$.
Under the second item in Assumtion \ref{aspt_topology}, the coupling linear equations in \eqref{eq_LE_temp}, we have $\frac{1}{|\mathbb{S}(\mathcal{R}_{\epsilon_k})|}\textbf{u}_{\epsilon_k} = \bar{\textbf{B}}_{i,\epsilon_k}\bar{\textbf{z}}_i + \sum_{j \in \mathcal{N}_{i,\epsilon_k}} \textbf{v}_{ij, \epsilon_k}^{(i)}  - \bar{\textbf{b}}_{i,\epsilon_k}$ if the following constraints are satisfied for all $j \in \mathcal{N}_{i,\epsilon_k}$ where $i \in \mathbb{S}(\mathcal{R}_{\epsilon_k})$.
\begin{subequations}\label{eq_LE_virtual_constraint}
\begin{gather}
\sum\limits_{l \in \mathcal{N}_{i,\epsilon_k}} \textbf{v}_{il, \epsilon_k}^{(i)} + \bar{\textbf{B}}_{i,\epsilon_k}\bar{\textbf{z}}_i - \bar{\textbf{b}}_{i,\epsilon_k} =  \sum\limits_{l \in \mathcal{N}_{j,\epsilon_k}} \textbf{v}_{jl, \epsilon_k}^{(j)} + \bar{\textbf{B}}_{j,\epsilon_k}\bar{\textbf{z}}_j - \bar{\textbf{b}}_{j,\epsilon_k},\\
\textbf{v}_{ij,\epsilon_k}^{(i)} + \textbf{v}_{ij,\epsilon_k}^{(j)} = \textbf{0}.
\end{gather}
\end{subequations}
Let $\mathcal{M}_i = \{\epsilon_k \in \mathbb{M}: i \in \mathbb{S}(\mathcal{R}_{\epsilon_k})\}$ for each agent $i \in \mathcal{V}$.
Define the stacked vector $\bar{\textbf{v}}_i = \textrm{col}\Bigl\{ \textrm{col}\Bigl\{\textbf{v}_{ij,\epsilon_k}^{(i)}: j \in \mathcal{N}_{i,\epsilon_k}\Bigr\}: \epsilon_k \in \mathcal{M}_i\Bigr\}$.
In addition, we define the stacked matrix $\textbf{A}_i$ and the stacked vector $\textbf{a}_i$ for each agent $i \in \mathcal{V}$ as
\begin{align*}
\textbf{A}_i &= \textrm{blkcol}\bigl\{\bar{\textbf{A}}_{i,k}: \mathbb{S}(\mathcal{R}_k) = \{i\}, \forall 1 \le k \le K\bigr\},\\
\textbf{a}_i &= \textrm{blkcol}\bigl\{\bar{\textbf{a}}_{i,k}: \mathbb{S}(\mathcal{R}_k) = \{i\}, \forall 1 \le k \le K\bigr\}.
\end{align*}
For simple notation, we assume $\textbf{A}_i = \textbf{0}$ and $\textbf{a}_i = \textbf{0}$ for agent $i \in \mathcal{V}$ if $\mathbb{S}(\mathcal{R}_k) \neq \{i\}, \forall 1 \le k \le K$, i.e., the agent $i$ does not know all nonzero blocks in any row partition of the coefficient matrix $\textbf{H}$.
From \eqref{eq_local_equation_full}, \eqref{eq_LE_temp} and \eqref{eq_LE_virtual_constraint}, we have $\Psi(\bar{\textbf{z}}_1, \bar{\textbf{z}}_2, \dots, \bar{\textbf{z}}_N) = \sum_{i \in \mathcal{V}}\Psi_i(\bar{\textbf{z}}_i, \bar{\textbf{v}}_i)$ where $\Psi_i(\bar{\textbf{z}}_i, \bar{\textbf{v}}_i) = \frac{1}{2} \bigl|\bigl|\textbf{A}_i\bar{\textbf{z}}_i - \textbf{a}_i\bigr|\bigr|^2 + \sum\limits_{\epsilon_k \in \mathcal{M}_i} \frac{|\mathbb{S}(\mathcal{R}_{\epsilon_k})|}{2}\bigl|\bigl|\sum\limits_{j \in \mathcal{N}_{i,\epsilon_k}} \textbf{v}_{ij, \epsilon_k}^{(i)} + \bar{\textbf{B}}_{i,\epsilon_k}\bar{\textbf{z}}_i - \bar{\textbf{b}}_{i,\epsilon_k}\bigr|\bigr|^2$.

Consider the following optimization problem:
\begin{equation}\label{eq_LE_equivalentproblem}
\min_{\bar{\textbf{z}}_i, \bar{\textbf{v}}_i, \forall i \in \mathcal{V}} \sum_{i \in \mathcal{V}} \Psi_i(\bar{\textbf{z}}_i, \bar{\textbf{v}}_i) \textrm{ s.t. } \eqref{eq_coupling_commonvariable}, \eqref{eq_LE_virtual_constraint} \forall \epsilon_k \in \mathbb{M}.
\end{equation}
We have the following result:
\begin{Lemma}\label{lm_equivalent}
Let $\textrm{col}\{\textrm{col}\{\bar{\textbf{z}}_i^{opt}, \bar{\textbf{v}}_i^{opt}\}: i \in \mathcal{V}\}$ be one optimal solution to the problem \eqref{eq_LE_equivalentproblem}. Then $\bar{\textbf{z}}_i^{opt}, \forall i \in \mathcal{V},$ is the stacked vector of the parts corresponding to the agent $i$ in a least square solution to the linear algebraic equation \eqref{eq_LE}.
\end{Lemma}
\begin{proof}
Lemma \ref{lm_equivalent} is proved by showing that $\textrm{col}\{\bar{\textbf{z}}_i^{opt}: i \in \mathcal{V}\}$ is one optimal solution to the problem \eqref{eq_LE_optimizationproblem_rewritten}.
We consider the contradiction. Assuming that there exists a set of vectors $\{\bar{\textbf{z}}_i^*: i \in \mathcal{V}\}$ satisfying constraint \eqref{eq_coupling_commonvariable} and $\Psi(\bar{\textbf{z}}_1^*, \bar{\textbf{z}}_2^*, \dots, \bar{\textbf{z}}_N^*) < \Psi(\bar{\textbf{z}}_1^{opt}, \bar{\textbf{z}}_2^{opt}, \dots, \bar{\textbf{z}}_N^{opt})$. For each $\epsilon_k \in \mathbb{M}$, we define $\textbf{u}_{\epsilon_k}^* = \sum_{i \in \mathbb{S}(\mathcal{R}_{\epsilon_k})} \left(\bar{\textbf{B}}_{i,\epsilon_k}\bar{\textbf{z}}_i^* - \bar{\textbf{b}}_{i,\epsilon_k}\right)$ and the set $\{\textbf{v}_{ij,\epsilon_k}^*: (i,j) \in \mathcal{E}_{\epsilon_k}\}$ satisfying \eqref{eq_LE_temp_virtual_var}.

Let $\bar{\textbf{v}}_i^* = \textrm{col}\bigl\{ \textrm{col}\bigl\{\textbf{v}_{ij,\epsilon_k}^* - \textbf{v}_{ji,\epsilon_k}^*: j \in \mathcal{N}_{i,\epsilon_k}\bigr\}: \epsilon_k \in \mathcal{M}_i\bigr\}$, $\forall i \in \mathcal{V}$. It is clear that $\textrm{col}\{\textrm{col}\{\bar{\textbf{z}}_i^*, \bar{\textbf{v}}_i^*\}: i \in \mathcal{V}\}$ is a feasible solution of the optimization problem \eqref{eq_LE_equivalentproblem}. Moreover, we have $\sum_{i \in \mathcal{V}} \Psi_i(\bar{\textbf{z}}_i^*, \bar{\textbf{v}}_i^*) = \Psi(\bar{\textbf{z}}_1^*, \bar{\textbf{z}}_2^*, \dots, \bar{\textbf{z}}_N^*) < \Psi(\bar{\textbf{z}}_1^{opt}, \bar{\textbf{z}}_2^{opt}, \dots, \bar{\textbf{z}}_N^{opt}) = \sum_{i \in \mathcal{V}} \Psi_i(\bar{\textbf{z}}_i^{opt}, \bar{\textbf{v}}_i^{opt})$. This contradicts to the fact that $\textrm{col}\{\textrm{col}\{\bar{\textbf{z}}_i^{opt}, \bar{\textbf{v}}_i^{opt}\}: i \in \mathcal{V}\}$ is one optimal solution to the problem \eqref{eq_LE_equivalentproblem}.
\end{proof}
\begin{Remark}
In APPENDIX-A, we provide an example to illustrate the list and sets in distributed setup and the reformulation of distributed linear algebraic equation (\ref{eq_local_equation_full}-\ref{eq_coupling_equation}) into the distributed optimization problem \eqref{eq_LE_equivalentproblem}.
\end{Remark}
\section{Distributed ADMM-based method}
In this section, we apply the proximal ADMM algorithm proposed in \cite{BingshengHe2015} to design a distributed method for solving the optimization problem \eqref{eq_LE_equivalentproblem}. The designed method is guaranteed to have exponential convergence.
\subsection{ADMM based solution}
Let $\textbf{x}_i = \textrm{col}\{\bar{\textbf{z}}_i,\bar{\textbf{v}}_i\}$ be the collection of all local variables of the agent $i \in \mathcal{V}$.
Let $\bar{n}_i$ be the length of the vector $\bar{\textbf{z}}_i$ for every agent $i \in \mathcal{V}$. We have $\bar{\textbf{z}}_i = \textrm{blkrow}\{\textbf{I}_{\bar{n}_i}, \textbf{0}\}\textbf{x}_i, \forall i \in \mathcal{V}$.
Define $\textbf{u}_{\epsilon_k}^{(i)}$ as the copy of the vector $\textbf{u}_{\epsilon_k}$  in \eqref{eq_LE_virtual_var} and the stacked vector $\bar{\textbf{u}}_i = \textrm{col}\{\textbf{u}_{\epsilon_k}^{(i)}: \epsilon_k \in \mathcal{M}_i\}$. Then we have $\bar{\textbf{u}}_i = \textbf{C}_i\textbf{x}_i - \textbf{c}_i$ where $\textbf{c}_i = \textrm{col}\{\bar{\textbf{b}}_{i,\epsilon_k}: \epsilon_k \in \mathcal{M}_{i}\}$ and\\
$\textbf{C}_i = \left[\begin{matrix}
\textrm{blkrow}\{\bar{\textbf{B}}_{i,\epsilon_k}^T: \epsilon_k \in \mathcal{M}_{i}\}\\
\textrm{blkdiag}\bigl\{\textrm{blkcol}\bigl\{\textbf{I}_{m_k}: j \in \mathcal{N}_{i,\epsilon_k}\bigr\}: \epsilon_k \in \mathcal{M}_i\bigr\}
\end{matrix}\right]^T$.
The local cost function $\Psi_i(\bar{\textbf{z}}_i, \bar{\textbf{v}}_i)$ of the agent $i \in \mathcal{V}$ has the quadratic form as \[\Psi_i(\bar{\textbf{z}}_i, \bar{\textbf{u}}_i, \bar{\textbf{v}}_i) = \frac{1}{2}\textbf{x}_i^T\textbf{Q}_i\textbf{x}_i + \textbf{q}_i^T\textbf{x}_i + \textbf{a}_i^T\textbf{a}_i + \textbf{c}_i^T\textbf{c}_i.\]
where $\textbf{Q}_i = \textrm{blkdiag}\Bigl\{\textbf{A}_i^T\textbf{A}_i, \textbf{0}\Bigr\} + \textbf{C}_i^T\textbf{C}_i$ and $\textbf{q}_i = \textrm{col}\{-\textbf{A}_i^T\textbf{a}_i, \textbf{0}\} + \textbf{C}_i^T\textbf{c}_i$.
It is easy to verify that $\textbf{Q}_i$ is a symmetric and positive semidefinite matrix and $\textbf{C}_i$ is a full row rank matrix, $\forall i \in \mathcal{V}$.

For two agents $i$ and $j$, where $(i,j) \in \mathcal{E}$ and $\epsilon_k \in \mathcal{M}_i \cap \mathcal{M}_j$, we assume that they choose the pair of matrices $\textbf{D}_{ij,\epsilon_k}$ and $\textbf{D}_{ji,\epsilon_k}$ such that if $\textbf{D}_{ij,\epsilon_k} = \textbf{I}$ then $\textbf{D}_{ji,\epsilon_k} = -\textbf{I}$ and if $\textbf{D}_{ij,\epsilon_k} = -\textbf{I}$ then $\textbf{D}_{ji,\epsilon_k} = \textbf{I}$ and vice versa.
Define the matrices $\textbf{E}_{ij}, \forall j \in \mathcal{N}_i$, such that
\begin{equation}
\textbf{E}_{ij}\textbf{x}_i = \textrm{col}\{\textbf{P}_{ij}\bar{\textbf{z}}_i, \textrm{col}\{\textbf{u}_{\epsilon_k}^{(i)}, \textbf{D}_{ij,\epsilon_k}\textbf{v}_{ij,\epsilon_k}^{(i)}: \epsilon_k \in \mathcal{M}_i \cap \mathcal{M}_j\}\}.
\end{equation}
It is clear that $\textbf{E}_{ij}\textbf{x}_i = \textbf{E}_{ji}\textbf{x}_j, \forall (i,j) \in \mathcal{E}$.
The optimization problem \eqref{eq_LE_equivalentproblem} can be rewritten in the following form:
\begin{subequations}\label{eq_distributed_optADMM}
\begin{align}
\min\limits_{\substack{\textbf{x}_i, \forall i \in \mathcal{V}\\ (\textbf{y}_{ij}, \textbf{y}_{ji}) \in \mathcal{Y}_{ij}, \forall (i,j) \in \mathcal{E}}}& \sum_{i \in \mathcal{V}}\left(\frac{1}{2}\textbf{x}_i^T\textbf{Q}_i\textbf{x}_i + \textbf{q}_i^T\textbf{x}_i\right)\\
\textrm{ s.t. }& \textbf{E}_{ij}\textbf{x}_i = \textbf{y}_{ij}, \forall j \in \mathcal{N}_i, \forall i \in \mathcal{V}.
\end{align}
\end{subequations}
where $\mathcal{Y}_{ij} = \{(\textbf{y}_{ij}, \textbf{y}_{ji}): \textbf{y}_{ij} = \textbf{y}_{ji}\}$ for all $(i,j) \in \mathcal{E}$.
The coupling constraints \eqref{eq_coupling_commonvariable} and \eqref{eq_LE_virtual_constraint} between two neighboring agents $i$ and $j$, where $(i,j) \in \mathcal{E}$, are replaced by the equality constraints in (\ref{eq_distributed_optADMM}b) combining the constraint $(\textbf{y}_{ij}, \textbf{y}_{ji}) \in \mathcal{Y}_{ij}$.
It is easy to verify that the problem \eqref{eq_distributed_optADMM} has the form of the convex optimization problem \eqref{eq_ADMMproblem} with the variable vectors $\textbf{x} = \textrm{col}\{\textbf{x}_1, \textbf{x}_2, \dots, \textbf{x}_N\}$, $\textbf{y} = \textrm{col}\{\textrm{col}\{\textbf{y}_{ij}, \textbf{y}_{ji}\}: (i,j) \in \mathcal{E}\}$, the cost functions $\Psi_x(\textbf{x}) = \sum_{i \in \mathcal{V}}\Bigl\{\frac{1}{2}\textbf{x}_i^T\textbf{Q}_i\textbf{x}_i + \textbf{q}_i^T\textbf{x}_i\Bigr\}$, $\Psi_y(\textbf{y}) = 0$, and
$\mathcal{Y} = \prod\limits_{(i,j) \in \mathcal{E}} \mathcal{Y}_{ij}$.
So, the optimal solution of the problem \eqref{eq_distributed_optADMM} can be found by using the proximal ADMM algorithm \eqref{eq_ADMM}.

Define $\boldsymbol{\lambda}_{ij}$, $\forall j \in \mathcal{N}_i$, as the dual variables corresponding to the equality constraints in (\ref{eq_distributed_optADMM}b) of the agent $i \in \mathcal{V}$. These dual variables can be stacked into a vector $\boldsymbol{\lambda} = \textrm{col}\{\textrm{col}\{\boldsymbol{\lambda}_{ij}, \boldsymbol{\lambda}_{ji}\}: (i,j) \in \mathcal{E}\}$.
The augmented Lagrangian function of the problem \eqref{eq_distributed_optADMM} is given by
\[\mathcal{L}_{\rho}(\textbf{x},\textbf{y},\boldsymbol{\lambda}) = \sum\limits_{i \in \mathcal{V}} \Bigl\{\frac{1}{2}\textbf{x}_i^T\textbf{Q}_i\textbf{x}_i + \textbf{q}_i^T\textbf{x}_i + \frac{\rho}{2}\sum\limits_{j \in \mathcal{N}_i}||\textbf{E}_{ij}\textbf{x}_i - \textbf{y}_{ij} - \frac{1}{\rho}\boldsymbol{\lambda}_{ij}||^2\Bigr\}.\]
Due to the separation of the augmented Lagrangian function $\mathcal{L}(\textbf{x},\textbf{y},\boldsymbol{\lambda})$, the equations (\ref{eq_ADMM}a) and (\ref{eq_ADMM}b) are equivalent to the equations \eqref{eq_updatelaw_primal_1} and \eqref{eq_updatelaw_primal_2}, respectively. Here, the matrix $\textbf{G} = \textrm{blkdiag}\{\textbf{G}_i: i \in \mathcal{V}\}$ where $\textbf{G}_i$ is a positive semi-definite matrix chosen by the agent $i \in \mathcal{V}$.
\begin{equation}\label{eq_updatelaw_primal_1}
\textbf{x}_i(s+1) = \arg\min\limits_{\textbf{x}_i} \left(\frac{1}{2}\textbf{x}_i^T\textbf{Q}_i\textbf{x}_i + \textbf{q}_i^T\textbf{x}_i + \frac{\rho}{2}\sum\limits_{j \in \mathcal{N}_i}||\textbf{E}_{ij}\textbf{x}_i - \textbf{y}_{ij}(s) - \frac{1}{\rho}\boldsymbol{\lambda}_{ij}(s)||^2 + \frac{1}{2}||\textbf{x}_i - \textbf{x}_i(s)||_{\textbf{G}_i}\right), \forall i \in \mathcal{V}.
\end{equation}
\begin{equation}\label{eq_updatelaw_primal_2}
(\textbf{y}_{ij},\textbf{y}_{ji})(s+1) = \arg\min\limits_{\textbf{y}_{ij} = \textbf{y}_{ji}} \left(||\textbf{E}_{ij}\textbf{x}_i(s+1) - \textbf{y}_{ij} - \frac{1}{\rho}\boldsymbol{\lambda}_{ij}(s)||^2 + ||\textbf{E}_{ji}\textbf{x}_j(s+1) - \textbf{y}_{ji} - \frac{1}{\rho}\boldsymbol{\lambda}_{ji}(s)||^2\right), \forall (i,j) \in \mathcal{E}.
\end{equation}
\subsection{Detailed algorithm}
Assume that $\textbf{G}_i$ is chosen such that the matrix $\hat{\textbf{Q}}_i$ (defined in (\ref{eq_updatelaw_primal_1_detailed}b)) is positive definite, the update \eqref{eq_updatelaw_primal_1} is equivalent to \eqref{eq_updatelaw_primal_1_detailed}, $\forall i \in \mathcal{V}$.
\begin{subequations}\label{eq_updatelaw_primal_1_detailed}
\begin{align}
\textbf{x}_i(s+1) &= -\hat{\textbf{Q}}_i^{-1}\hat{\textbf{q}}_i(s),\\
\textrm{where } \hat{\textbf{Q}}_i &= \textbf{Q}_i + \textbf{G}_i + \rho\sum_{j \in \mathcal{N}_i} \textbf{E}_{ij}^T\textbf{E}_{ij},\\
\hat{\textbf{q}}_i(s) &= \textbf{q}_i - \textbf{G}_i\textbf{x}_i(s) - \sum\limits_{j \in \mathcal{N}_i}\textbf{E}_{ij}^T\left(\rho\textbf{y}_{ij}(s) + \boldsymbol{\lambda}_{ij}(s)\right).
\end{align}
\end{subequations}
Consider the optimization problem in \eqref{eq_updatelaw_primal_2}. The KKT conditions of this problem consist of
\begin{gather}
\textbf{E}_{ij}\textbf{x}_i(s+1) - \textbf{y}_{ij}(s+1) - \frac{1}{\rho}\boldsymbol{\lambda}_{ij}(s) + \boldsymbol{\mu} = \textbf{0}, \label{eq_temp1}\\
\textbf{E}_{ji}\textbf{x}_j(s+1) - \textbf{y}_{ji}(s+1) - \frac{1}{\rho}\boldsymbol{\lambda}_{ji}(s) - \boldsymbol{\mu} = \textbf{0}, \label{eq_temp2}\\
\textbf{y}_{ij}(s+1) = \textbf{y}_{ji}(s+1), \label{eq_temp3}
\end{gather}
where $\boldsymbol{\mu}$ is the dual variable corresponding to the constraint $\textbf{y}_{ij} = \textbf{y}_{ji}$.
Adding two equations \eqref{eq_temp1} and \eqref{eq_temp2}, we obtain $\textbf{y}_{ij}(s+1) + \textbf{y}_{ji}(s+1) = \textbf{E}_{ij}\textbf{x}_i(s+1) - \frac{1}{\rho}\boldsymbol{\lambda}_{ij}(s) + \textbf{E}_{ji}\textbf{x}_j(s+1) - \textbf{y}_{ji}(s+1) - \frac{1}{\rho}\boldsymbol{\lambda}_{ji}(s)$.
From \eqref{eq_temp3}, we have 
\begin{equation}\label{eq_updatelaw_primal_2_detailed}
\textbf{y}_{ij}(s+1) = \frac{1}{2}\left(\textbf{E}_{ij}\textbf{x}_i(s+1) - \frac{1}{\rho}(\boldsymbol{\lambda}_{ij}(s) + \textbf{E}_{ji}\textbf{x}_j(s+1) - \frac{1}{\rho}\boldsymbol{\lambda}_{ji}(s)\right),
\end{equation}
for all $(i,j) \in \mathcal{E}$.
According to (\ref{eq_ADMM}c), the update of the dual variables $\boldsymbol{\lambda}_{ij}, \forall (i,j) \in \mathcal{E},$ are given as
\begin{equation}\label{eq_updatelaw_dual}
\boldsymbol{\lambda}_{ij}(s+1) = \boldsymbol{\lambda}_{ij}(s) - \rho\left(\textbf{E}_{ij}\textbf{x}_i(s+1) - \textbf{y}_{ij}(s+1)\right).
\end{equation}
We provide Algorithm \ref{alg_proposed_control} to highlight the implementation of our designed method for one agent. It is easy to verify that this algorithm requires only local information for each agent except the common penalty parameter $\rho$ as global information. However, this information can be easily achieved by the agreement among agents. In the iteration update (\ref{eq_updatelaw_primal_1_detailed}-\ref{eq_updatelaw_dual}), ever agent $i \in \mathcal{V}$ uses only information belonging to itself or received from its neighbors in $\mathcal{N}_i$.
\begin{algorithm}[htb]
\begin{algorithmic}[1]
\BState \emph{Agree a common positive penalty parameter} $\rho > 0$
\BState \emph{Choose} $\textbf{G}_i \succeq 0$ such that $\hat{\textbf{Q}}_i \succ 0$ (given in (\ref{eq_updatelaw_primal_1_detailed}b))
\BState \emph{Initiate:} $s = 0, \textbf{x}_i(0) = \textbf{0}$ and $\textbf{y}_{ij}(0) = \textbf{0}, \forall j \in \mathcal{N}_i$. 
\BState \emph{Repeat until convergence:}
\State $\quad$ $\textbf{x}_i(s+1) \gets$ \eqref{eq_updatelaw_primal_1_detailed},
\State $\quad$ Send $\textbf{E}_{ij}\textbf{x}_i(s+1) - \frac{1}{\rho}\boldsymbol{\lambda}_{ij}(s)$ to all $j \in \mathcal{N}_i$,
\State $\quad$ Receive $\textbf{E}_{ji}\textbf{x}_j(s+1) - \frac{1}{\rho}\boldsymbol{\lambda}_{ji}(s)$ from all $j \in \mathcal{N}_i$,
\State $\quad$ $\textbf{y}_i(s+1) \gets$ \eqref{eq_updatelaw_primal_2_detailed}, $\forall j \in \mathcal{N}_i$,
\State $\quad$ $\boldsymbol{\lambda}_{ij}(s+1) \gets$ \eqref{eq_updatelaw_dual}, $\forall j \in \mathcal{N}_i$,
\State $\quad$ $s \gets s + 1$,
\BState \emph{Finish algorithm:} $\textbf{z}_i^{output} = \textrm{blkrow}\{\textbf{I}_{\bar{n}_i}, \textbf{0}\}\textbf{x}_i(s)$
\end{algorithmic}
\caption{Distributed method for agent $i \in \mathcal{V}$ to find $\bar{\textbf{z}}_{\mathcal{S}_i}^{opt}$.}\label{alg_proposed_control}
\end{algorithm}

Let $\textbf{w}_i(s) = \textrm{col}\Bigl\{ \textbf{x}_i(s), \textrm{col}\bigl\{\textbf{y}_{ij}(s), \boldsymbol{\lambda}_{ij}(s): j \in \mathcal{N}_i\bigr\}\Bigr\}$ be the stacked vector corresponding to the estimated solution of the agent $i \in \mathcal{V}$ at the iteration step $s$ in the update (\ref{eq_updatelaw_primal_1_detailed}-\ref{eq_updatelaw_dual}).
Based on Theorem \ref{th_ADMM}, we have the following theorem:
\begin{Theorem}\label{th_main}
Let every agent $i \in \mathcal{V}$ apply Algorithm 1. The estimated solution $\textbf{w}(s) = \textrm{col}\{\textbf{w}_i(s): i \in \mathcal{V}\}$ converges exponentially to $\textbf{w}^{\infty}$ where $\textbf{w}^{\infty} = \textrm{col}\{\textbf{w}_i^{\infty}: i \in \mathcal{V}\}$ is one optimal solution of the optimization problem \eqref{eq_distributed_optADMM}.
\end{Theorem}
\begin{proof}
The first item in Theorem \ref{th_ADMM} implies there exists a converged point $\textbf{w}^{\infty}$, i.e., $\lim_{s \rightarrow \infty} \textbf{w}(s) = \textbf{w}^{\infty}$, for arbitrarily initial point $\textbf{w}(0)$ where $\textbf{w}(s) = \textrm{col}\{\textbf{w}_i(s): i \in \mathcal{V}\}$ and $\textbf{w}^{\infty} = \textrm{col}\{\textbf{w}_i^{\infty}: i \in \mathcal{V}\}$. The second and third items of Theorem \ref{th_ADMM} correspond to the convergence of the equality constraint and the cost function value. They guarantee that the converged point $\textbf{w}^{\infty}$ satisfies all constraints and attain the minimum cost function value of the problem \eqref{eq_distributed_optADMM}. So $\textbf{w}^{\infty}$ is one optimal solution to the optimization problem \eqref{eq_distributed_optADMM}.

According to (\ref{eq_updatelaw_primal_1_detailed}-\ref{eq_updatelaw_dual}), there exist a constant matrix $\textbf{M}$ and a constant vector $\textbf{m}$ such that
\begin{equation}\label{eq_nonhomo}
\textbf{w}(s+1) = \textbf{M}\textbf{w}(s) + \textbf{m}.
\end{equation}
The detailed formulations of $\textbf{M}$ and $\textbf{m}$ are presented in this paper since they are not necessary in the analysis.

Define $\overline{\textbf{w}}(s) = \textbf{w}(s+1) - \textbf{w}(s)$, we have $\overline{\textbf{w}}(s) \rightarrow \textbf{0}$ as $s \rightarrow \infty$ with all initial setup $\textbf{w}(0)$.
From \eqref{eq_nonhomo}, we have \[\overline{\textbf{w}}(s+1) = \textbf{M}\overline{\textbf{w}}(s).\]
By choosing $\textbf{w}(0) = \textbf{0}$, we have $\overline{\textbf{w}}(s) = \textbf{M}^s\textbf{m}$. So, $\lim_{s \rightarrow \infty} \textbf{M}^s\textbf{m} = \textbf{0}$.
Let $\textbf{v}_{(\sigma)}$ be the eigenvector of the matrix $\textbf{M}$ corresponding to  the eigenvalue $\sigma$, i.e., $\textbf{M}\textbf{v}_{(\sigma)} = \sigma\textbf{v}_{(\sigma)}$. By choosing $\textbf{w}(0) = \textbf{v}_{(\sigma)}$, we have $\overline{\textbf{w}}(s) = (\sigma - 1)\sigma^s\textbf{v}_{(\sigma)} + \textbf{M}^s\textbf{m}$. The fact $\lim_{s \rightarrow \infty}\overline{\textbf{w}}(s) = \lim_{s \rightarrow \infty}(\sigma - 1)\sigma^s\textbf{v}_{(\sigma)} = \textbf{0}$ can be guaranteed if $|\sigma| < 1$ or $\sigma = 1$.
That means, for every eigenvalue $\sigma$ of the matrix $\textbf{M}$, if $\sigma \neq 1$ then $|\sigma| < 1$.
So, if the matrix $\textbf{M}$ does not have an eigenvalue at $1$, $\textbf{w}(s+1)$ converges to $\textbf{w}^{\infty}$ exponentially because $\textbf{w}(s+1) - \textbf{w}^{\infty} = \textbf{M}\left(\textbf{w}(s) - \textbf{w}^{\infty}\right)$.

Assume that the matrix $\textbf{M}$ has an eigenvalue at $1$ and its algebraic multiplicity is bigger than its geometric multiplicity. That means there exists a vector $\textbf{u}_{(1)}$ such that $(\textbf{M} - \textbf{I})\textbf{u}_{(1)} = \textbf{v}_{(1)}$.
Let $\textbf{w}(0) = \textbf{u}_{(1)}$, we have $\overline{\textbf{w}}(s) = \textbf{v}_{(1)} + \textbf{M}^s\textbf{m}$. This contradicts the asymptotic convergence to zero vector of $\overline{\textbf{w}}(s)$.
So, if $\textbf{M}$ has an eigenvalue at $1$, its algebraic multiplicity is equal to its geometric multiplicity.
This fact implies there exists a nonsingular matrix $\textbf{T}$ such that $\textbf{M} = \textbf{T}\left[\begin{matrix}\textbf{I}_k & \textbf{0}\\ \textbf{0} & \textbf{J} \end{matrix}\right]\textbf{T}^{-1}$ where $k$ is the algebraic multiplicity corresponding to the eigenvalue $1$ of the matrix $\textbf{M}$ and the matrix $\textbf{J}$ has all eigenvalues lying within the unit open disk.
Since $\textbf{w}(s+1) - \textbf{w}^{\infty} = \textbf{M}\left(\textbf{w}(s) - \textbf{w}^{\infty}\right)$, we have
\[\textbf{w}(s+l) - \textbf{w}^{\infty} = \textbf{T}\left[\begin{matrix} \textbf{I}_k & \textbf{0}\\ \textbf{0} & \textbf{J}^l \end{matrix}\right]\textbf{T}^{-1}\left(\textbf{w}(s) - \textbf{w}^{\infty}\right).\]
The asymptotic convergence of the estimated solution $\textbf{w}(s)$ to the converged point $\textbf{w}^{\infty}$ implies
\[\textbf{0} = \lim_{l \rightarrow \infty}\left(\textbf{w}(s+l) - \textbf{w}^{\infty}\right) = \textbf{T}\left[\begin{matrix}\textbf{I}_k & \textbf{0}\\ \textbf{0} & \textbf{0} \end{matrix}\right]\textbf{T}^{-1}\left(\textbf{w}(s) - \textbf{w}^{\infty}\right)\]
because $\textbf{J}^l \rightarrow \textbf{0}$ as $l \rightarrow \infty$.
So we have
\begin{align*}
\textbf{w}(s+1) - \textbf{w}^{\infty} &= \textbf{M}(\textbf{w}(s) - \textbf{w}^{\infty}) - \textbf{0}\\
&= \left(\textbf{M} - \textbf{T}\left[\begin{matrix}\textbf{I}_k & \textbf{0}\\ \textbf{0} & \textbf{0} \end{matrix}\right]\textbf{T}^{-1}\right)\left(\textbf{w}(s) - \textbf{w}^{\infty}\right)\\
&= \textbf{T}\left[\begin{matrix} \textbf{0}_{k \times k} & \textbf{0}\\ \textbf{0} & \textbf{J} \end{matrix}\right]\textbf{T}^{-1}\left(\textbf{w}(s) - \textbf{w}^{\infty}\right)\\
&= \textbf{T}\left[\begin{matrix} \textbf{0}_{k \times k} & \textbf{0}\\ \textbf{0} & \textbf{J}^{s+1} \end{matrix}\right]\textbf{T}^{-1}\left(\textbf{w}(0) - \textbf{w}^{\infty}\right).
\end{align*}
Since all eigenvalues in the matrix $\textbf{J}$ are within the unit open disk, $\textbf{w}(s) - \textbf{w}^{\infty}$ converges exponentially to the origin.
\end{proof}
According to Theorem \ref{th_main}, $\bar{\textbf{z}}_i(s) = \textrm{blkrow}\{\textbf{I}_{\bar{n}_i}, \textbf{0}\}\textbf{x}_i(s)$ converges exponentially to $\bar{\textbf{z}}_i^{\infty} = \textrm{blkrow}\{\textbf{I}_{\bar{n}_i}, \textbf{0}\}\textbf{x}_i^{\infty}, \forall i \in \mathcal{V}$.
Lemma \ref{lm_equivalent} guarantees that the stacked vector $\bar{\textbf{z}}_i^{\infty}, \forall i \in \mathcal{V},$ consists of parts corresponding to the agent $i$ in one least square solution of the linear algebraic equation \eqref{eq_LE}.
\section{Numerical simulations}
\subsection{Distributed setup with rows and columns decomposition}
This part considers a distributed linear algebraic equation over a network of $24$ agents, which can be represented in the form of (\ref{eq_local_equation_full}-\ref{eq_coupling_equation}). Assume that each agent has $20$ local variables and there is one coupling row partition across all agents. That means $\bar{\textbf{z}}_i \in \mathbb{R}^{20}$ for all $i \in \{1, 2, \dots, 24\}$ and $\mathbb{M} = \{\epsilon_1\}$, $\mathcal{R}_{\epsilon_1} = \{1, 2, \dots, 24\}$.
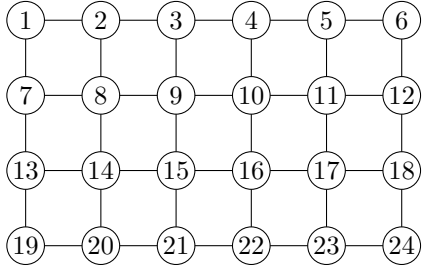
\begin{figure}[htb]
\centering
\begin{tikzpicture}
\draw (0,0) grid (5,3);
\node[draw,circle,minimum size=0.5cm,inner sep=0pt, fill=white, text=black] at (0,3) {$1$};
\node[draw,circle,minimum size=0.5cm,inner sep=0pt, fill=white, text=black] at (1,3) {$2$};
\node[draw,circle,minimum size=0.5cm,inner sep=0pt, fill=white, text=black] at (2,3) {$3$};
\node[draw,circle,minimum size=0.5cm,inner sep=0pt, fill=white, text=black] at (3,3) {$4$};
\node[draw,circle,minimum size=0.5cm,inner sep=0pt, fill=white, text=black] at (4,3) {$5$};
\node[draw,circle,minimum size=0.5cm,inner sep=0pt, fill=white, text=black] at (5,3) {$6$};
\node[draw,circle,minimum size=0.5cm,inner sep=0pt, fill=white, text=black] at (0,2) {$7$};
\node[draw,circle,minimum size=0.5cm,inner sep=0pt, fill=white, text=black] at (1,2) {$8$};
\node[draw,circle,minimum size=0.5cm,inner sep=0pt, fill=white, text=black] at (2,2) {$9$};
\node[draw,circle,minimum size=0.5cm,inner sep=0pt, fill=white, text=black] at (3,2) {$10$};
\node[draw,circle,minimum size=0.5cm,inner sep=0pt, fill=white, text=black] at (4,2) {$11$};
\node[draw,circle,minimum size=0.5cm,inner sep=0pt, fill=white, text=black] at (5,2) {$12$};
\node[draw,circle,minimum size=0.5cm,inner sep=0pt, fill=white, text=black] at (0,1) {$13$};
\node[draw,circle,minimum size=0.5cm,inner sep=0pt, fill=white, text=black] at (1,1) {$14$};
\node[draw,circle,minimum size=0.5cm,inner sep=0pt, fill=white, text=black] at (2,1) {$15$};
\node[draw,circle,minimum size=0.5cm,inner sep=0pt, fill=white, text=black] at (3,1) {$16$};
\node[draw,circle,minimum size=0.5cm,inner sep=0pt, fill=white, text=black] at (4,1) {$17$};
\node[draw,circle,minimum size=0.5cm,inner sep=0pt, fill=white, text=black] at (5,1) {$18$};
\node[draw,circle,minimum size=0.5cm,inner sep=0pt, fill=white, text=black] at (0,0) {$19$};
\node[draw,circle,minimum size=0.5cm,inner sep=0pt, fill=white, text=black] at (1,0) {$20$};
\node[draw,circle,minimum size=0.5cm,inner sep=0pt, fill=white, text=black] at (2,0) {$21$};
\node[draw,circle,minimum size=0.5cm,inner sep=0pt, fill=white, text=black] at (3,0) {$22$};
\node[draw,circle,minimum size=0.5cm,inner sep=0pt, fill=white, text=black] at (4,0) {$23$};
\node[draw,circle,minimum size=0.5cm,inner sep=0pt, fill=white, text=black] at (5,0) {$24$};
\end{tikzpicture}
\caption{Communication graph: nodes represent agents and edges represent available communication among agents.}\label{fig_simulation}
\end{figure}
Let the communication graph be given in Fig. \ref{fig_simulation} and $\bar{\textbf{B}}_{i,\epsilon_1} = \left[\begin{matrix}\textbf{0}_{5 \times 15} & \textbf{I}_{5}\end{matrix}\right], \forall i \in \{1, 2, \dots, 24\}$. Every vector $\bar{\textbf{b}}_{i,\epsilon_1}$ is chosen randomly.
In addition, we assume that the first $5$ elements in local variable vectors $\bar{\textbf{z}}_i, \forall i \in \{1, 2, \dots, 24\}$, are common variables. In other word, $\textbf{P}_{ij} = \textbf{P}_{ji} = \left[\begin{matrix}\textbf{I}_5 & \textbf{0}_{5 \times 15}\end{matrix}\right]$ for all pairs $(i,j)$ of neighboring agents. 
We use MATLAB to test the effectiveness of our designed method in finding least square solutions to one hundred numerical examples.
In each example, we choose $\textbf{A}_i \in \mathbb{R}^{20 \times 20}, \textbf{a}_i \in \mathbb{R}^{20}$ randomly, $\forall i \in \{1, 2, \dots, 24\}$. The ADMM-based method (\ref{eq_updatelaw_primal_1_detailed}-\ref{eq_updatelaw_dual}) is run with the selection of $\rho = 1$.
\begin{figure}[htb]
\centering
\includegraphics[width=0.5\textwidth]{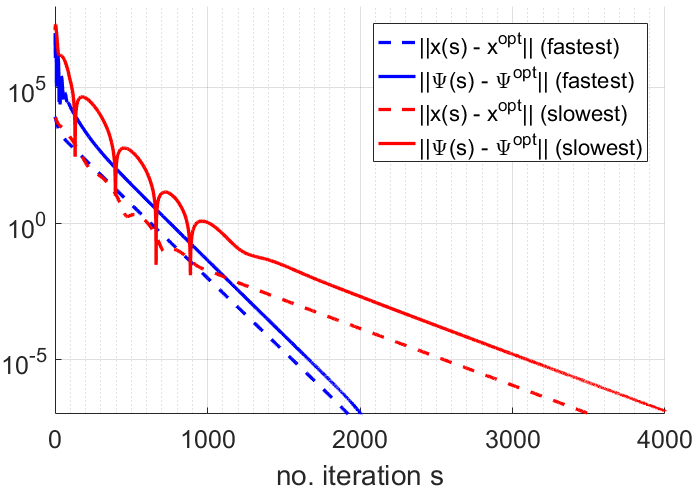}
\caption{The convergence to the least square solution of the proposed algorithm for the distributed linear algebraic equation with rows and columns decomposition.}\label{fig_result}
\end{figure}
Fig. \ref{fig_result} shows the simulation results of two examples having the fastest (blue lines) and the slowest (red lines) convergence rates. In this figure, the solid lines represent the convergence of the estimated solution in the update (\ref{eq_updatelaw_primal_1_detailed}-\ref{eq_updatelaw_dual}) to one optimal solution $\textbf{x}^{opt}$ of the problem \eqref{eq_distributed_optADMM}. The dashed lines illustrate to the evolution of the function $|\Psi(s) - \Psi^{opt}|$ where $\Psi(s) = \sum_{i = 1}^{24}\left(\textbf{x}_i(s)^T\textbf{Q}_i\textbf{x}_i(s) + \textbf{q}_i^T\textbf{x}_i(s) + \textbf{a}_i^T\textbf{a}_i + \textbf{c}_i^T\textbf{c}_i\right)$ and $\Psi^{opt}$ is the optimal cost value.
From the simulation results, it is easy to verify the exponential convergence of the estimated solution to the optimal one. A practical solution can be achieved after about $2000$ iterations. In addition, we measure the computation time required by each agent in every iteration step. The maximum time required by one agent for $2000$ iterations is $0.5$ second.
\subsection{Distributed setup with rows decomposition}
In this part, we compare our designed method with the existing methods proposed in Wang et. al. 2019 \cite{XuanWang2019} and Yang et. al. 2020 \cite{TaoYang2020}.
 Assume that the tested linear algebraic equation system has the form of \eqref{eq_distributedLE_row} and the communication network is given in Fig. \ref{fig_graph_example}. Let each agent $i$ know the row partitions $\textbf{A}_i$ and $\textbf{b}_i$.
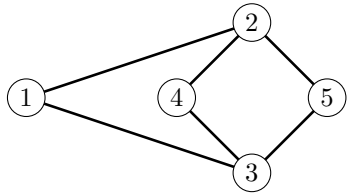
\begin{figure}[htb]
\centering
\begin{tikzpicture}
\node[] at (0,0) (n1){};
\node[] at (3,1) (n2){};
\node[] at (3,-1) (n3){};
\node[] at (2,0) (n4){};
\node[] at (4,0) (n5){};
\draw[-,{line width=1pt},black] (n2)--(n1);
\draw[-,{line width=1pt},black] (n3)--(n1);
\draw[-,{line width=1pt},black] (n2)--(n4);
\draw[-,{line width=1pt},black] (n3)--(n4);
\draw[-,{line width=1pt},black] (n2)--(n5);
\draw[-,{line width=1pt},black] (n3)--(n5);
\node[draw,circle,minimum size=0.5cm,inner sep=0pt, fill=white, text=black] at (0,0){$1$};
\node[draw,circle,minimum size=0.5cm,inner sep=0pt, fill=white, text=black] at (3,1){$2$};
\node[draw,circle,minimum size=0.5cm,inner sep=0pt, fill=white, text=black] at (3,-1){$3$};
\node[draw,circle,minimum size=0.5cm,inner sep=0pt, fill=white, text=black] at (2,0){$4$};
\node[draw,circle,minimum size=0.5cm,inner sep=0pt, fill=white, text=black] at (4,0){$5$};
\end{tikzpicture}
\caption{Communication graph of $5$ agents.}\label{fig_graph_example}
\end{figure}
\begin{equation}\label{eq_distributedLE_row}
\textbf{A} = \left[\begin{matrix} \textbf{A}_{1}\\ \textbf{A}_{2}\\ \textbf{A}_{3}\\ \textbf{A}_{4}\\ \textbf{A}_{5} \end{matrix}\right], \textbf{b} = \left[\begin{matrix} \textbf{b}_{1}\\ \textbf{b}_{2}\\ \textbf{b}_{3}\\ \textbf{b}_{4}\\ \textbf{b}_{5}\end{matrix}\right]
\end{equation}

We first test the convergence of three distributed algorithms in small linear algebraic equation systems $(\textbf{A}, \textbf{b})$ where $\textbf{A}_i \in \mathcal{R}^{1 \times 4}$. Fig. \ref{fig_result1} and Fig. \ref{fig_result2} correspond with the case $(\textbf{A}, \textbf{b}) = (\textbf{A}^{(1)}, \textbf{b}^{(1)})$ and $(\textbf{A}, \textbf{b}) = (\textbf{A}^{(2)}, \textbf{b}^{(2)})$, respectively. Note that, the linear algebraic equation $(\textbf{A}^{(1)}, \textbf{b}^{(1)})$ has a unique solution while the linear algebraic equation $(\textbf{A}^{(2)}, \textbf{b}^{(2)})$ has infinity solutions.
\begin{equation}
\textbf{A}^{(1)} = \left[\begin{matrix} 1 & 2 & 1 & 1\\ 2 & -1 & -1 & 1\\ 1 & -2 & 4 & -1\\ -1 & -0.6 & 0.4 & 1.8\\ 2 & 2 & -2 & 1\end{matrix}\right],
\textbf{b}^{(1)} = \left[\begin{matrix} 10\\ 20\\ 15\\ 17\\ 11 \end{matrix}\right]
\textrm{ and }
\textbf{A}^{(2)} = \left[\begin{matrix} 1 & 2 & 1 & 1\\ 1 & 1.4 & -1.6 & 2.8\\ 3 & 3. & -3.6 & 3.8\\ -1 & -0.6 & 0.4 & 1.8\\ 2 & 2 & -2 & 1\end{matrix}\right],
\textbf{b}^{(2)} = \left[\begin{matrix} 10\\ 20\\ 15\\ 17\\ 11 \end{matrix}\right]
\end{equation}
\begin{figure}[htb]
\centering
\includegraphics[width=0.75\textwidth]{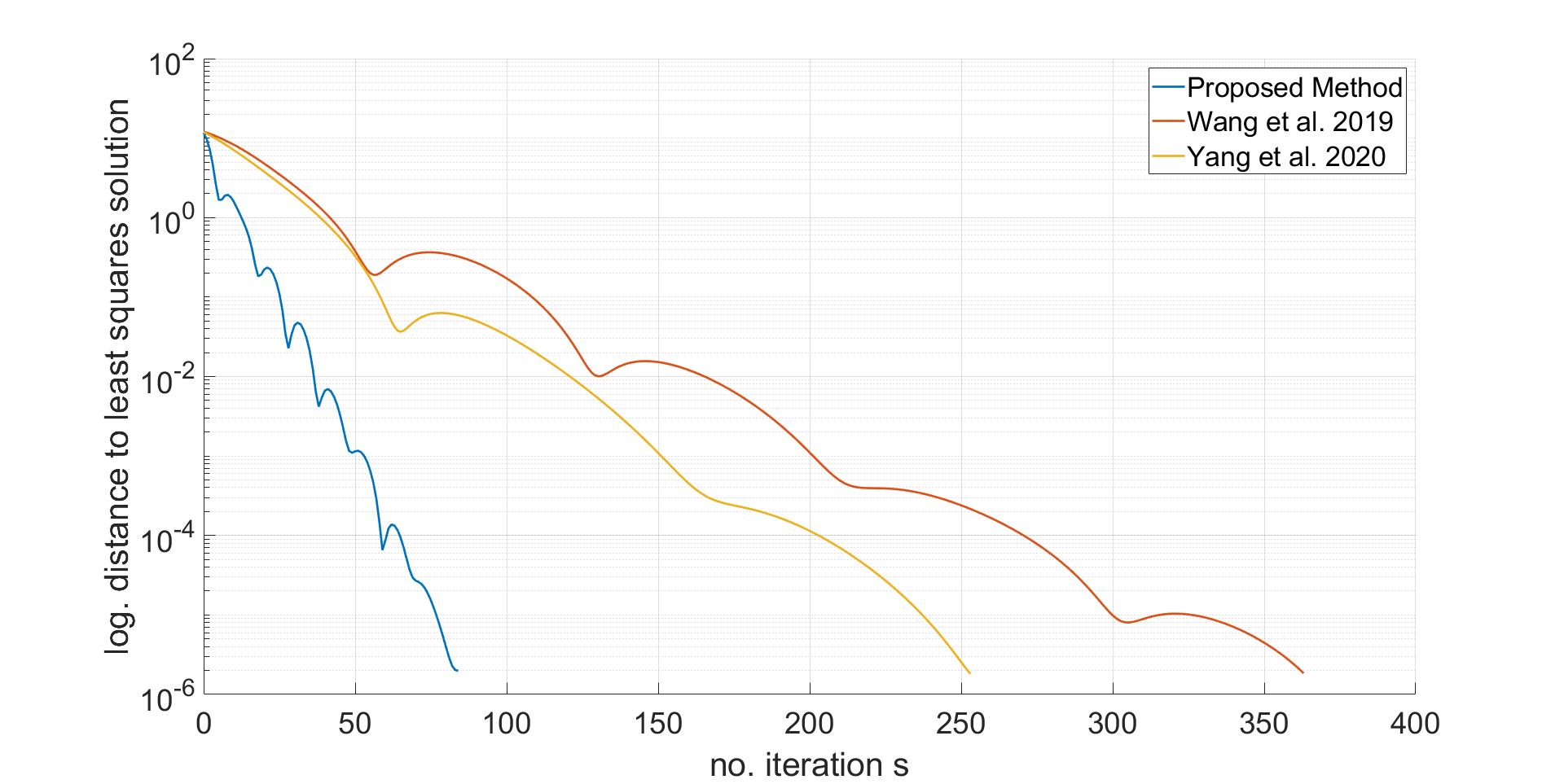}
\caption{The convergence to the least square solution of three distributed algorithms for the distributed linear algebraic equation $(\textbf{A}^{(1)}, \textbf{b}^{(1)})$ with $5$ agents.}\label{fig_result1}
\end{figure}
\begin{figure}[htb]
\centering
\includegraphics[width=0.75\textwidth]{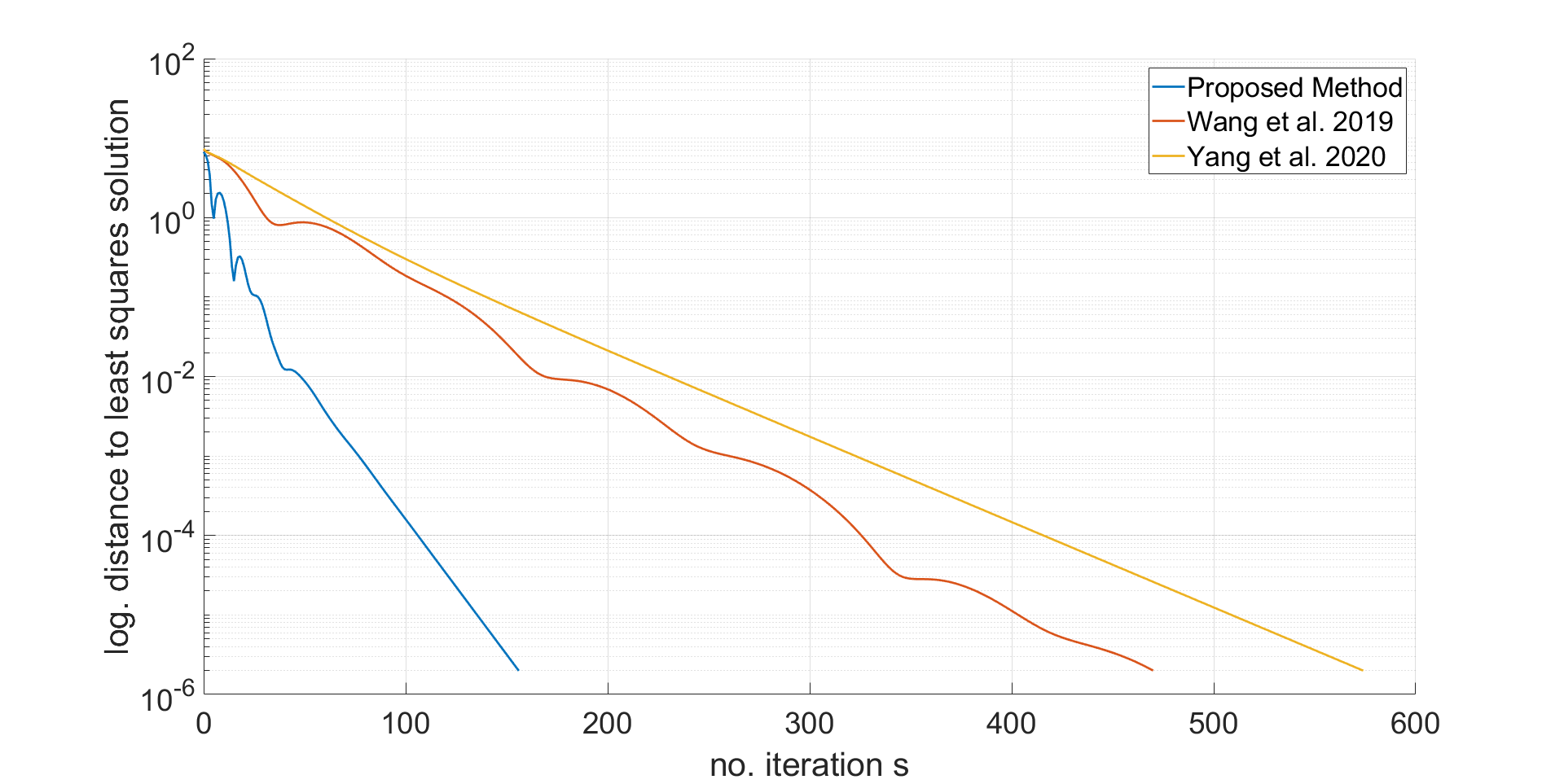}
\caption{The convergence to the least square solution of three distributed algorithms for the distributed linear algebraic equation $(\textbf{A}^{(2)}, \textbf{b}^{(2)})$ with $5$ agents.}\label{fig_result2}
\end{figure}
In both cases, our proposed method requires significantly less iterations than two others.
Finally, we compare the effectiveness of three distributed methods for solving a larger linear algebraic equation system. Fig. \ref{fig_result3} presents the numerical results when the matrix $\textbf{A}_i$ and the vector $\textbf{b}_i$ are chosen randomly such that $\textbf{A}_i \in \mathbb{R}^{60 \times 30}$ and $\textbf{b} \in \mathbb{R}^{60}$.
\begin{figure}[htb]
\centering
\includegraphics[width=0.75\textwidth]{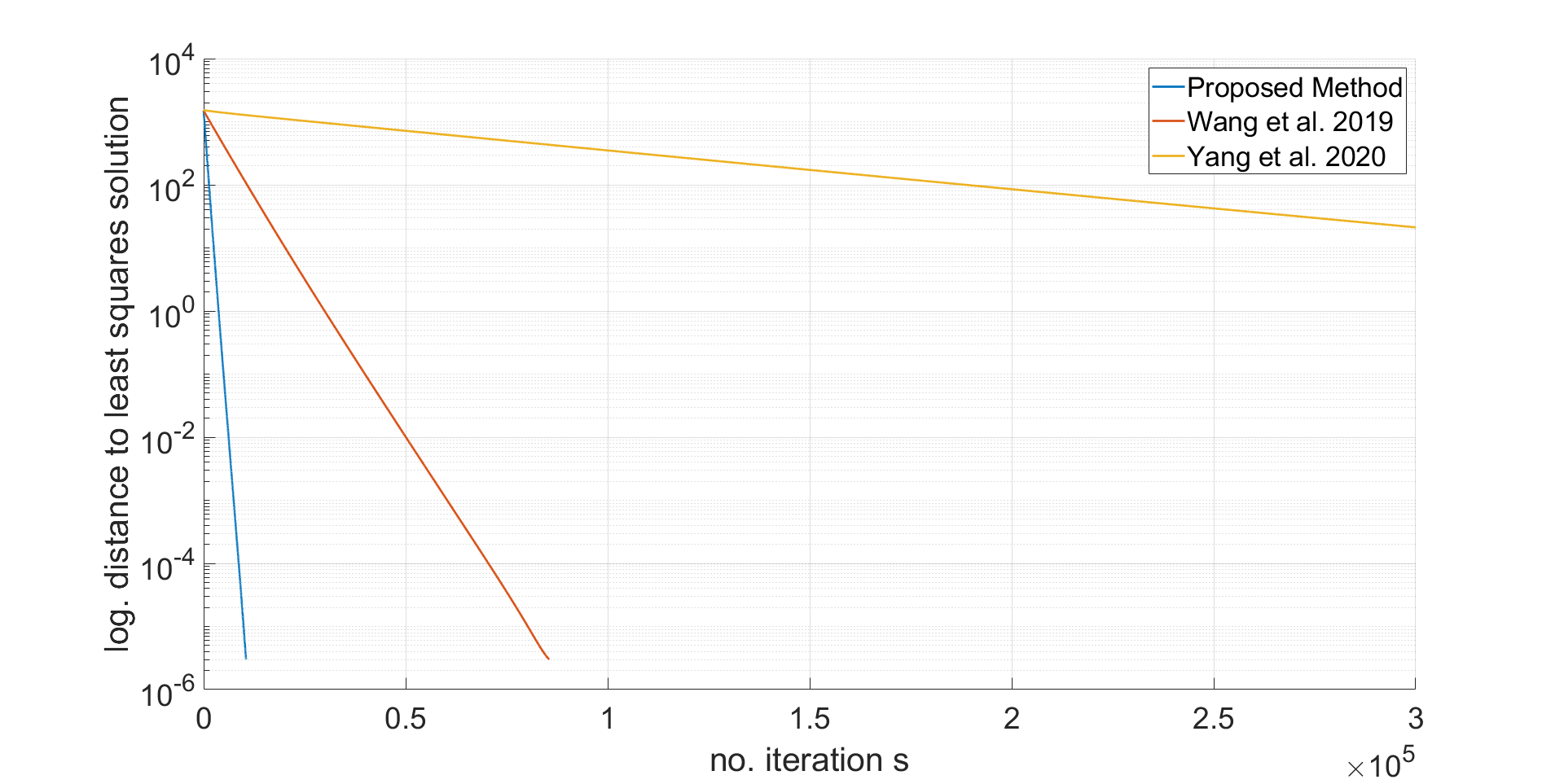}
\caption{The convergence to the least square solution of three distributed algorithms for larger distributed linear algebraic equation with $5$ agents.}\label{fig_result3}
\end{figure}
In this case, the advanced convergence of our proposed method is shown more clearly.
\section{Conclusion}
This paper considers a linear algebraic equation $\textbf{H}\textbf{z} = \textbf{h}$ when rows and columns of the coefficient matrix $\textbf{H}$ are dispensed across a network of agents. By introducing some virtual variables, the problem of finding one least square solution to the considered linear algebraic equation is reformulated into a distributed optimization problem. Then, we apply a proximal ADMM algorithm to design a distributed method for agents to solve the formulated optimization problem cooperatively. 
The estimated solution of every agent converges exponentially to its corresponding parts in one exact least square solution of the considered linear algebraic equation.
\appendix
\subsection{Illustrate the reformulation the problem of finding the least square solution to distributed optimization problem}
To illustrate the list and sets defined above, we consider an example as shown in Fig. \ref{fig_example}. First, we observe that $\mathcal{B}_1^{(1)} = \{\textbf{H}_{11}, \textbf{H}_{13}\}$, $\mathcal{B}_1^{(2)} = \mathcal{B}_1^{(3)} = \mathcal{B}_1^{(4)} = \mathcal{B}_1^{(5)} = \emptyset$; $\mathcal{B}_2^{(1)} = \mathcal{B}_2^{(3)} = \mathcal{B}_2^{(5)} = \emptyset$, $\mathcal{B}_2^{(2)} = \{\textbf{H}_{21}\}$, $\mathcal{B}_2^{(4)} = \{\textbf{H}_{24}\}$; $\mathcal{B}_3^{(1)} = \mathcal{B}_3^{(2)} = \mathcal{B}_3^{(4)} = \mathcal{B}_3^{(5)} = \emptyset$, $\mathcal{B}_3^{(3)} = \{\textbf{H}_{32}, \textbf{H}_{33}\}$; $\mathcal{B}_4^{(1)} = \mathcal{B}_4^{(5)} = \emptyset$, $\mathcal{B}_4^{(2)} = \{\textbf{H}_{41}\}$, $\mathcal{B}_4^{(3)} = \{\textbf{H}_{43}\}$, $\mathcal{B}_4^{(4)} = \{\textbf{H}_{44}\}$; $\mathcal{B}_5^{(1)} = \mathcal{B}_5^{(2)} = \mathcal{B}_5^{(3)} = \mathcal{B}_5^{(4)} = \emptyset$, $\mathcal{B}_5^{(5)} = \{\textbf{H}_{51}, \textbf{H}_{53}\}$; $\mathcal{B}_6^{(1)} = \mathcal{B}_6^{(2)} = \mathcal{B}_6^{(3)} = \mathcal{B}_6^{(4)} = \emptyset$, $\mathcal{B}_6^{(5)} = \{\textbf{H}_{61}, \textbf{H}_{62}\}$.
The index lists corresponding to the row and column partitions are given as: $\mathcal{R}_1 = \{1, 0, 1, 0\}$, $\mathcal{R}_2 = \{2, 0, 0, 4\}$, $\mathcal{R}_3 = \{0, 3, 3, 0\}$, $\mathcal{R}_4 = \{2, 0, 3, 4\}$, $\mathcal{R}_5 = \{5, 0, 5, 0\}$, $\mathcal{R}_6 = \{5, 5, 0, 0\}$, and $\mathcal{C}^1 = \{1, 2, 0, 2, 5, 5\}$, $\mathcal{C}^2 = \{0, 0, 3, 0, 0, 5\}$, $\mathcal{C}^3 = \{1, 0, 3, 3, 5, 0\}$, $\mathcal{C}^4 = \{0, 4, 0, 4, 0, 0\}$.
\begin{figure}[htb]
\centering
\begin{tikzpicture}
\node[] at (0,1.5) {$\textbf{H} = \left[\begin{matrix}
\textbf{H}_{11} & \textbf{0} & \textbf{H}_{13} & \textbf{0}\\
\textbf{H}_{21} & \textbf{0} & \textbf{0} & \textbf{H}_{24}\\
\textbf{0} & \textbf{H}_{32} & \textbf{H}_{33} & \textbf{0}\\
\textbf{H}_{41} & \textbf{0} & \textbf{H}_{43} & \textbf{H}_{44}\\
\textbf{H}_{51} & \textbf{0} & \textbf{H}_{53} & \textbf{0}\\
\textbf{H}_{61} & \textbf{H}_{62} & \textbf{0} & \textbf{0}
\end{matrix}\right]$};
\node[] at (0,-0.2) {a. Coefficient matrix};
\end{tikzpicture}
\begin{tikzpicture}
\node[] at (-1.5,2.8) {Agent $1$ knows $\textbf{H}_{11}, \textbf{H}_{13}$;};
\node[] at (-1.5,2.3) {Agent $2$ knows $\textbf{H}_{21}, \textbf{H}_{41}$;};
\node[] at (-1.15,1.8) {Agent $3$ knows $\textbf{H}_{32}, \textbf{H}_{33}, \textbf{H}_{43}$;};
\node[] at (-1.5,1.3) {Agent $4$ knows $\textbf{H}_{24}, \textbf{H}_{44}$;};
\node[] at (-0.75,0.8) {Agent $5$ knows $\textbf{H}_{51}, \textbf{H}_{53}, \textbf{H}_{61}, \textbf{H}_{62}$;};
\node[] at (-1.2,0.1) {b. Distribution of data};
\end{tikzpicture}
\caption{Example of a distributed linear algebraic equation over a network of $5$ agents.}\label{fig_example}
\end{figure}
The sets corresponding to these lists are: $\mathbb{S}(\mathcal{R}_1) = \{1\}$, $\mathbb{S}(\mathcal{R}_2) = \{2, 4\}$, $\mathbb{S}(\mathcal{R}_3) = \{3\}$, $\mathbb{S}(\mathcal{R}_4) = \{2, 3, 4\}$, $\mathbb{S}(\mathcal{R}_5) = \mathbb{S}(\mathcal{R}_6) = \{5\}$, and $\mathbb{S}(\mathcal{C}^1) = \{1, 2, 5\}$, $\mathbb{S}(\mathcal{C}^2) = \{3, 5\}$, $\mathbb{S}(\mathcal{C}^3) = \{1, 3, 5\}$, $\mathbb{S}(\mathcal{C}^4) = \{4\}$.
Consider the distributed linear algebraic equation in Fig. \ref{fig_example} and assume that the communication graph of the multiagent network is given in Fig. \ref{fig_graph_example}. It can be verified that Assumption \ref{aspt_topology} is satisfied.
We have the unknown vector $\textbf{z} = \textrm{col}\{\textbf{z}_1, \textbf{z}_2, \textbf{z}_3, \textbf{z}_4\}$ and the coefficient vector $\textbf{h} = \textrm{col}\{\textbf{h}_1, \textbf{h}_2, \textbf{h}_3, \textbf{h}_4, \textbf{h}_5, \textbf{h}_6\}$.
From agent-index sets corresponding to column partitions, we have $\bar{\textbf{z}}_1 = \textrm{col}\{\textbf{z}_1^{(1)}, \textbf{z}_3^{(1)}\}$, $\bar{\textbf{z}}_2 = \textbf{z}_1^{(2)}$, $\bar{\textbf{z}}_3 = \textrm{col}\{\textbf{z}_2^{(3)}, \textbf{z}_3^{(3)}\}$, $\bar{\textbf{z}}_4 = \textbf{z}_4^{(4)}$, $\bar{\textbf{z}}_5 = \textrm{col}\{\textbf{z}_1^{(5)}, \textbf{z}_2^{(5)}, \textbf{z}_3^{(5)}\}$.
The matrices corresponding in the coupled constraints in \eqref{eq_coupling_commonvariable} are $\textbf{P}_{12} = \textrm{col}\{\textbf{I}_{n_1}, \textbf{0}_{n_1 \times n_3}\}$, $\textbf{P}_{21} = \textbf{P}_{25} = \textbf{I}_{n_1}$, $\textbf{P}_{52} = \textrm{col}\{\textbf{I}_{n_1}, \textbf{0}_{n_1 \times (n_2+n_3)}\}$, $\textbf{P}_{13} = \textrm{col}\{\textbf{0}_{n_3 \times n_1}, \textbf{I}_{n_3}\}$, $\textbf{P}_{31} = \textrm{col}\{\textbf{0}_{n_3 \times n_2}, \textbf{I}_{n_3}\}$, $\textbf{P}_{35} = \textbf{I}_{n_2+n_3}$ and $\textbf{P}_{53} = \textrm{col}\{\textbf{0}_{(n_2+n_3) \times n_1}, \textbf{I}_{n_2+n_3}\}$.

For the row partition $k \in \{1, 3, 5, 6\}$, the set $\mathbb{S}(\mathcal{R}_k)$ has only one agent. Then we have $\bar{\textbf{A}}_{1,1} = \textrm{blkcol}\{\textbf{H}_{11}, \textbf{H}_{13}\}$, $\bar{\textbf{a}}_{1,1} = \textbf{h}_1$; $\bar{\textbf{A}}_{3,3} = \textrm{blkcol}\{\textbf{H}_{32}, \textbf{H}_{33}\}$, $\bar{\textbf{a}}_{3,3} = \textbf{h}_3$; $\bar{\textbf{A}}_{5,5} = \textrm{blkcol}\{\textbf{H}_{51}, \textbf{0}, \textbf{H}_{53}\}$, $\bar{\textbf{a}}_{5,5} = \textbf{h}_5$ and $\bar{\textbf{A}}_{5,6} = \textrm{blkcol}\{\textbf{H}_{61}, \textbf{H}_{62}, \textbf{0}\}$, $\bar{\textbf{a}}_{5,6} = \textbf{h}_6$.
The coupling row partition set is $\mathbb{M} = \{2, 4\}$. As $\mathbb{S}(\mathcal{R}_2) = \{2, 4\}$, two agents $2$ and $4$ know respectively $\bar{\textbf{B}}_{2,2}, \bar{\textbf{b}}_{2,2}$ and $\bar{\textbf{B}}_{4,2}, \bar{\textbf{b}}_{4,2}$ where $\bar{\textbf{B}}_{2,2} = \textbf{H}_{21}$, $\bar{\textbf{B}}_{4,2} = \textbf{H}_{24}$ and $\bar{\textbf{b}}_{2,2} + \bar{\textbf{b}}_{4,2} = \textbf{h}_2$. Agent $2$ has the local variables $\textbf{u}_2^{(2)}, \textbf{v}_{24,2}$ and agent $4$ has the local variables $\textbf{u}_2^{(4)}, \textbf{v}_{42,2}$.
Similarly, for the row partition $3$, agents $2$ knows $\textbf{u}_4^{(2)}, \textbf{v}_{24,4}, \bar{\textbf{B}}_{2,4}, \bar{\textbf{b}}_{2,4}$, agents $3$ knows $\textbf{u}_4^{(3)}, \textbf{v}_{34,4}, \bar{\textbf{B}}_{3,4}, \bar{\textbf{b}}_{3,4}$ and agents $4$ knows $\textbf{u}_4^{(4)}, \textbf{v}_{42,4}, \textbf{v}_{43,4}, \bar{\textbf{B}}_{4,4}, \bar{\textbf{b}}_{4,4}$. In which $\bar{\textbf{B}}_{2,4} = \textbf{H}_{41}$, $\bar{\textbf{B}}_{3,4} = \textrm{blkcol}\{\textbf{0}, \textbf{H}_{43}\}$, $\bar{\textbf{B}}_{4,4} = \textbf{H}_{44}$ and $\bar{\textbf{b}}_{2,4} + \bar{\textbf{b}}_{3,4} + \bar{\textbf{b}}_{4,4} = \textbf{h}_4$.

Finally, we summarize the local information for each agent as follows.
Agent $1$ knows $\bar{\textbf{z}}_1$, $\textbf{P}_{12}$, $\textbf{P}_{13}$, $\textbf{A}_1 = \bar{\textbf{A}}_{1,1}$ and $\textbf{a}_1 = \bar{\textbf{a}}_{1,1}$;
Agent $2$ knows $\bar{\textbf{z}}_2$, $\textbf{P}_{21}$, $\textbf{P}_{25}$, $\bar{\textbf{u}}_2 = \textrm{col}\{\textbf{u}_2^{(2)}, \textbf{u}_4^{(2)}\}$, $\bar{\textbf{v}}_2 = \textrm{col}\{\textbf{v}_{24,2}, \textbf{v}_{24,4}\}$, $\bar{\textbf{B}}_{2,2}$, $\bar{\textbf{b}}_{2,2}$, $\bar{\textbf{B}}_{2,4}$, $\bar{\textbf{b}}_{2,4}$;
Agent $3$ knows $\bar{\textbf{z}}_3$, $\textbf{A}_3 = \bar{\textbf{A}}_{3,3}$, $\textbf{a}_3 = \bar{\textbf{a}}_{3,3}$, $\textbf{P}_{31}$, $\textbf{P}_{35}$, $\bar{\textbf{u}}_3 = \textbf{u}_4^{(3)}$, $\bar{\textbf{v}}_3 = \textbf{v}_{34,2}$, $\bar{\textbf{B}}_{3,4}$, $\bar{\textbf{b}}_{3,4}$;
Agent $4$ knows $\bar{\textbf{z}}_4$, $\bar{\textbf{u}}_4 = \textrm{col}\{\textbf{u}_2^{(4)}, \textbf{u}_4^{(4)}\}$, $\bar{\textbf{v}}_4 = \textrm{col}\{\textbf{v}_{42,2}, \textbf{v}_{42,4}, \textbf{v}_{43,4}\}$, $\bar{\textbf{B}}_{4,2}$, $\bar{\textbf{b}}_{4,2}$, $\bar{\textbf{B}}_{4,4}$, $\bar{\textbf{b}}_{4,4}$;
Agent $5$ knows $\bar{\textbf{z}}_5$, $\textbf{P}_{52}$, $\textbf{P}_{53}$, $\textbf{A}_5 = \textrm{blkcol}\{\bar{\textbf{A}}_{5,5}, \bar{\textbf{A}}_{5,6}\}$ and $\textbf{a}_5 = \textrm{blkcol}\{\bar{\textbf{a}}_{5,5}, \bar{\textbf{a}}_{5,6}\}$.
\bibliographystyle{IEEEtran}
\bibliography{mylib}
\end{document}